\documentclass[journal, final]{IEEEtran}
\usepackage[T1]{fontenc}
\usepackage[latin9]{inputenc}
\usepackage{amsthm}
\usepackage{amsmath}
\usepackage{amssymb}
\usepackage{algpseudocode}
\usepackage{bm}
\usepackage{booktabs}
\usepackage{url}
\usepackage{cite}
\usepackage{color}
\usepackage{float}
\usepackage{graphicx}
\usepackage{listings}
\usepackage{mathrsfs}
\usepackage{multicol}
\usepackage{multirow}
\usepackage{stfloats}
\usepackage{diagbox}

\makeatletter

\floatstyle{ruled}
\newfloat{algorithm}{tbp}{loa}
\providecommand{\algorithmname}{Algorithm}
\floatname{algorithm}{\protect\algorithmname}
\usepackage{tikz}
\usetikzlibrary{calc}

\theoremstyle{plain}
\newtheorem{lem}{\protect\lemmaname}
\theoremstyle{plain}
\newtheorem{cor}{\protect\corollaryname}
\theoremstyle{plain}
\newtheorem{prop}{\protect\propositionname}
\theoremstyle{plain}
\newtheorem{thm}{\protect\theoremname}

\usepackage[acronym]{glossaries}
\newcommand{\newac}{\newacronym}
\newcommand{\ac}{\gls}
\newcommand{\Ac}{\Gls}
\newcommand{\acpl}{\glspl}

\newac{speb}{SPEB}{square position error bound}
\newac[plural=EFIMs,firstplural=Fisher information matrices (EFIMs)]{efim}{EFIM}{Fisher information matrix}
\newac{ne}{NE}{Nash equilibrium}
\newac{mse}{MSE}{mean squared error}
\newac{toa}{TOA}{time-of-arrival}
\newac{snr}{SNR}{signal-to-noise ratio}
\newac{lan}{LAN}{local area network}
\newac{psd}{PSD}{positive semidefinite}
\newac{pd}{PD}{positive definite}
\newac{wrt}{w.r.t.}{with respect to}
\newac{lhs}{L.H.S.}{left hand side}
\newac{wp1}{w.p.1}{with probability 1}
\newac{kkt}{KKT}{Karush-Kuhn-Tucker}
\newac{wlog}{w.l.o.g.}{without loss of generality}
\newac{mle}{MLE}{maximum likelihood estimation}
\newac{gps}{GPS}{Global Positioning System}
\newac{rssi}{RSSI}{received signal strength indicator}
\newac{mimo}{MIMO}{multiple-input multiple-output}
\newac{csi}{CSI}{channel state information}
\newac{fdd}{FDD}{frequency division duplexing}
\newac{ms}{MS}{mobile station}
\newac{bs}{BS}{base station}
\newac{d2d}{D2D}{device-to-device}
\newac{slnr}{SLNR}{signal-to-interference-leakage-and-noise-ratio}
\newac{ula}{ULA}{uniform linear antenna array}
\newac{pas}{PAS}{power angular spectrum}
\newac{mmse}{MMSE}{minimum mean square error}
\newac{zf}{ZF}{zero-forcing}
\newac{rzf}{RZF}{regularized zero-forcing}
\newac{as}{AS}{angular spread}
\newac{aod}{AOD}{angle of departure}
\newac{iid}{i.i.d.}{independent and identically distributed} 
\newac{sinr}{SINR}{signal-to-interference-and-noise ratio}
\newac{tdd}{TDD}{time-division duplex}
\newac{rvq}{RVQ}{random vector quantization}
\newac{rhs}{R.H.S.}{right hand side}
\newac{mrc}{MRC}{maximum ratio combining}
\newac{cdf}{CDF}{cumulative distribution function}
\newac{a.s.}{a.s.}{almost surely}
\newac{los}{LOS}{line-of-sight}
\newac{jsdm}{JSDM}{joint spatial division and multiplexing}
\newac{map}{MAP}{maximum a posteriori}
\newac{klt}{KLT}{Karhunen-Lo\`eve Transform}
\newac{lbe}{LBE}{link bargaining equilibrium}
\newac{se}{SE}{Stackelberg equilibrium}
\newac{uav}{UAV}{unmanned aerial vehicle}
\newac{nlos}{NLOS}{non-line-of-sight}
\newac{pdf}{PDF}{probability density function}
\newac{em}{EM}{expectation-maximization}
\newac{knn}{KNN}{$k$-nearest neighbor}
\newac{svd}{SVD}{singular value decomposition}
\newac{nmf}{NMF}{non-negative matrix factorization}
\newac{umf}{UMF}{unimodality-constrained matrix factorization}
\newac{rmse}{RMSE}{rooted mean squared error}
\newac{olos}{OLOS}{obstructed line-of-sight}
\newac{mmw}{mmW}{millimeter wave}
\newac{ber}{BER}{bit error rate}
\newac{rss}{RSS}{received signal strength}
\newac{lp}{LP}{linear program}
\newac{ufw}{U-FW}{unimodal Frank-Wolfe}
\newac{utf}{UTF}{unimodality-constrained tensor factorization}
\newac{fw}{FW}{Frank-Wolfe}
\newac{iot}{IoT}{Internet-of-Things}
\newac{mae}{MAE}{mean absolute error}
\newac{crb}{CRB}{Cram\'er-Rao bound}
\newac{aoa}{AoA}{angle of arrival}
\newac{wcl}{WCL}{weighted centroid localization}
\newac[plural=GNSS,firstplural=global navigation satellite systems (GNSS)]{gnss}{GNSS}{global navigation satellite system}
\newac{gsm}{GSM}{Global System for Mobile Communications}
\newac{imu}{IMU}{inertial measurement unit}
\newac{rbf}{RBF}{radial basis function}
\newac{msf}{MSF}{multi-sensor fusion}
\newac{lidar}{LiDAR}{light detection and ranging}
\newac{glrt}{GLRT}{generalized likelihood ratio test}
\newac{sdr}{SDR}{software-defined radio}
\newac{ap}{AP}{access point}
\newac{lte}{LTE}{Long-Term Evolution}
\newac{raim}{RAIM}{receiver autonomous integrity monitoring}
\newac{pvt}{PVT}{position, velocity and time}
\newac{npl}{NPL}{Neyman-Pearson lemma}
\newac{ekf}{EKF}{extended Kalman filter}
\newac{tdoa}{TDOA}{time difference of arrival}
\newac{dop}{DOP}{dilution of precision}
\newac[plural=SOPs,firstplural=signals of opportunity (SOPs)]{sop}{SOP}{signals of opportunity}
\newac[plural=LBS,firstplural=location-based services (LBS)]{lbs}{LBS}{location-based service}
\newac{llm}{LLM}{large language model}
\newac{ssid}{SSID}{service set identifier}
\newac{bssid}{BSSID}{basic service set identifier}
\newac{mac}{MAC}{medium access control}
\newac[plural=POIs,firstplural=points of interest (POIs)]{poi}{POI}{point of interest}
\newac{uwb}{UWB}{ultra-wideband}
\newac{sbas}{SBAS}{satellite-based augmentation system}
\newac{rtt}{RTT}{round-trip time}
\newac{icmp}{ICMP}{Internet Control Message Protocol}
\newac{vpn}{VPN}{virtual private network}
\newac{roc}{ROC}{receiver-operating characteristic}
\newac{rtk}{RTK}{real-time kinematic}
\newac{rinex}{RINEX}{Receiver INdependent EXchange format}
\newac{nmea}{NMEA}{National Marine Electronics Association}

\makeatother

\providecommand{\corollaryname}{Corollary}
\providecommand{\lemmaname}{Lemma}
\providecommand{\propositionname}{Proposition}
\providecommand{\theoremname}{Theorem}

\title{Coordinated Position Falsification Attacks and Countermeasures for Location-Based Services}

\author{
Wenjie~Liu,~\IEEEmembership{Graduate~Student~Member,~IEEE,}
and~Panos~Papadimitratos,~\IEEEmembership{Fellow,~IEEE}
\thanks{W. Liu and P. Papadimitratos are with the Networked Systems Security Group, KTH Royal Institute of Technology, 114 28 Stockholm, Sweden.}
\thanks{Corresponding author: Wenjie Liu (e-mail: \textit{wenjieli@kth.se}).}
\thanks{This work was supported in part by the KAW Foundation and the China Scholarship Council. We thank the Jammertest 2024 organizers for a live \ac{gnss} jamming/spoofing/meaconing test environment, the former NSS member Dr. Marco Spanghero for support with simulations, smartphone antenna, and circuit components, and the National Academic Infrastructure for Supercomputing in Sweden (NAISS) for computational resources.}
}

\begin{document}

\maketitle              % typeset the header of the contribution

\begin{abstract}
With the rise of \ac{lbs} applications that rely on terrestrial and satellite infrastructures (e.g., \ac{gnss} and crowd-sourced Wi-Fi, Bluetooth, cellular, and IP databases) for positioning, ensuring their integrity and security is paramount. However, we demonstrate that these applications are susceptible to low-cost attacks (less than \$50), including Wi-Fi spoofing combined with \ac{gnss} jamming, as well as more sophisticated coordinated location spoofing. These attacks manipulate position data to control or undermine \ac{lbs} functionality, leading to user scams or service manipulation. Therefore, we propose a countermeasure to detect and thwart such attacks by utilizing readily available, redundant positioning information from off-the-shelf platforms. Our method extends the \ac{raim} framework by incorporating opportunistic information, including data from onboard sensors and terrestrial infrastructure signals, and, naturally, \ac{gnss}. We theoretically show that the fusion of heterogeneous signals improves resilience against sophisticated adversaries on multiple fronts. Experimental evaluations show the effectiveness of the proposed scheme in improving detection accuracy by 62\% at most compared to baseline schemes and restoring accurate positioning. 
\end{abstract}


\glsresetall

\section{Introduction}
\Ac{lbs} applications, integrated into daily life, depend on positioning derived from terrestrial and satellite infrastructures, such as cellular networks (3/4/5G), Wi-Fi, Bluetooth, and \ac{gnss} (e.g., \ac{gps}). For example, people use Google Maps for navigation or Uber for ride-hailing and food delivery services to their home or office. Netflix relies on \ac{lbs} to implement region locks. Accurate positioning is essential for the proper functionality and quality of these services. 

Recent studies and reports \cite{PaaKjeIntTho:C18,Wan:J16,Fan:J25,WanWanWanNik:J18,EryPap:C22,YilCukEmi:J23} have revealed real-world vulnerabilities in \ac{lbs} applications, which have even led to large-scale scams. Several attack instances involve generating false signals to manipulate positioning results used by \ac{lbs}. For example, the ``ghost driver'' is automatically assigned to ride-sharing passengers (e.g., using Uber) \cite{Wan:J16,Fan:J25}. This position manipulation allows the adversary to take subtle detours or simulate a pickup and trip completion, thereby carrying out a taxi fare scam. Game players of Pok{\'e}mon GO or users of Google Maps, Waze, Foursquare, etc., can fake the current position of the smartphone to gain non-compliant rewards \cite{PaaKjeIntTho:C18,EryPap:C22}. Scooter sharing services \cite{YilCukEmi:J23}, a popular urban transportation method, are often managed using geofencing, an \ac{lbs} feature that virtually restricts the riding area or dictates billing zones. However, position manipulation can break the geofencing, causing scooters to seemingly appear within the designated area while they are actually far outside. Public transport systems might use e-ticketing \cite{MicJatZibRaz:C24} and charge passengers based on travel distance; however, attackers can travel for free by crafting fake positions. 

A straightforward and increasingly serious threat to \ac{lbs} security is \ac{gnss} spoofing \cite{PsiHumSta:J16,SheWonCheChe:C20,AltMukKam:J23,YanEstVor:C23}. Attackers broadcast counterfeit satellite signals, often with higher strength but in the correct format. This fools receivers into locking onto the adversarial signals instead of the legitimate ones, thereby compromising position and time. Once considered exotic, \ac{gnss} spoofing attacks are now feasible using replays \cite{LenSpaPap:C22} or even open-source simulators. Wi-Fi geolocation, widely used in urban and indoor environments to complement \ac{gnss}, is also vulnerable to manipulation: Rogue Wi-Fi \acpl{ap} pose multiple security threats \cite{TipRasPopCap:C09,VanPie:C14,AloEll:J16,Skylift2016,HanXioSheLu:C22,HanXioSheWei:J24}. In our context, attackers broadcast Wi-Fi beacons (potentially downloaded or collected elsewhere) using consumer-grade Wi-Fi routers or low-cost Wi-Fi chips (e.g., ESP8266 \cite{Skylift2016}) to manipulate Wi-Fi-based positioning. Additionally, cellular-based positioning, which relies on signals from \acpl{bs}, is also a potential weak point of \ac{lbs}. Attackers can replay signals or deploy rogue \acpl{bs} to mislead positioning \cite{ShaBorParSei:C18}. For IP geolocation methods (GeoIP) such as using \ac{rtt} of \ac{icmp} ping, attackers can relay, use transparent proxies, or \ac{vpn} to control positioning results \cite{AbdMatVan:C17,KohDia:C22}. Android mock location API, designed for developers testing \ac{lbs}, provides a legal yet easy way to simulate false positions. 

Existing detections for \ac{gnss} spoofing, rogue Wi-Fi \acpl{ap}, and more broadly secure positioning usually rely upon specialized hardware or complex algorithms, which restricts their practicality. For instance, a multi-antenna array is needed to calculate the \ac{aoa} for \ac{gnss} spoofing detection \cite{SchRadCamFoo:J16,TanXieHuaLi:J25}, and Wi-Fi \ac{csi} fingerprints require certain network interfaces to measure \cite{LinGaoLiDon:C20,YanYanYanSon:J22}. However, most smartphones cannot even monitor the carrier phase of \ac{gnss} signals \cite{Gps:J24}. Furthermore, some recent proposals thwart attacks based on the assumption that some other infrastructure is out of reach of the adversary, e.g., \cite{LiuPap:J25a} detects \ac{gnss} spoofing while presuming that no Wi-Fi signal is adversarial. Hence, their detections fail in the presence of cellular jamming and rogue Wi-Fi \acpl{ap}. 

To highlight a rather challenging situation, we consider coordinated attacks on \ac{gnss} and other wireless signals. Specifically, we implement two new attacks. Due to the lack of encryption and authentication in Wi-Fi beacons, attackers can exploit this vulnerability by combining jamming \ac{gnss} signals with Wi-Fi spoofing. This forces the \ac{lbs} to rely solely on fake Wi-Fi beacons broadcast by the attacker for localization. Another attack involves coordinating the manipulation of position estimates derived from multiple infrastructures simultaneously. 

In response to these challenges, our solution leverages multiple opportunistic sources of ranging and motion information for attack detection and mitigation. We propose an enhanced \ac{raim} framework. Unlike previous work that often assumes a trusted signal source, we assume that any wireless signal contributing to the device's position estimation can potentially be attacked to manipulate the position. Since removing all available benign signals is considerably difficult for an attacker, it is reasonable to assume that one or more benign subsets often remain available. By cross-validating opportunistic ranging information from \ac{gnss} and terrestrial network infrastructures, our approach is compatible and can complement hardware fingerprinting or signal processing level detection \cite{Del:J24,FanYueXuHsu:J23,LinGaoLiDon:C20}, and independently improves the detection and mitigation of \ac{lbs} position manipulation scams. The scheme works mainly in two stages. The first stage generates subsets of the opportunistic ranging information. From these subsets, temporary position estimates and their associated uncertainties are calculated. Strategies employed, including subset sampling and \ac{llm}-based \ac{ap} name matching, improve both the efficiency and accuracy of our algorithms. Then, the second stage performs position fusion to detect \ac{gnss} and Wi-Fi spoofing using temporary estimates and cross-validates them. The algorithm identifies inconsistencies in temporary estimates indicative of spoofing attacks. 

In our earlier work \cite{LiuPap:C24,LiuPap:C24b}, we extended \ac{raim} to detect \ac{gnss} spoofing or rogue Wi-Fi \ac{ap} attacks and evaluated the detection in two simplified scenarios. A follow-up study \cite{LiuPap:C25b} proposed two new attacks on real-world \ac{lbs} with detections. However, due to legal restrictions on outdoor \ac{gnss} spoofing, these evaluations largely relied on simulations. Furthermore, few existing efforts have developed a testbed capable of supporting multimodal evaluation. In this paper, we address this gap by designing a secure and isolated over-the-air testbed that supports outdoor testing of multiple attacks on wireless signals. We also present a more comprehensive analysis of the detection theory, computational complexity, and experiments. 

The novelty and contribution are summarized as follows. We implement coordinated attacks at low cost on the most popular location providers and illustrate real-world \ac{lbs} position manipulations\footnote{\url{https://doi.org/10.5281/zenodo.15437800}}. Then, we develop a \ac{raim}-based framework for detecting and mitigating threats from these coordinated position attacks. By integrating opportunistic information from multiple sources, our approach provides a likelihood function against \ac{gnss} spoofing, rogue Wi-Fi \acpl{ap}, etc. In particular, in this paper, unlike earlier work, we propose a position recovery method and show its effectiveness even when the device is under attack. We prove its practicality without assuming any trusted source. We demonstrate the effectiveness of our approach through real-world evaluations, showing improved accuracy and reliability in detecting attacks in different scenarios. 

The remainder of this paper is organized as follows. Section~\ref{sec:backgr} provides background knowledge of \ac{lbs}, \ac{gnss} spoofing, and relevant attacks. Section~\ref{sec:sysmod} presents our system model and adversary. Section~\ref{sec:attlbs} shows how to launch attacks on \ac{lbs}. Section~\ref{sec:prosch} proposes our countermeasure. Section~\ref{sec:experi} evaluates the proposed scheme with baseline methods, and Section~\ref{sec:relwor} reviews related work about detection. Finally, Section~\ref{sec:conclu} concludes the paper. 

\section{Background}
\label{sec:backgr}
\subsection{Location-Based Service}
\Ac{lbs} has reshaped numerous aspects of our daily lives, offering personalized and context-conscious experiences. We discuss in detail the wide variety of applications of \ac{lbs} in different domains: Indoor navigation systems provide users with seamless navigation from outdoor to indoor environments, as \ac{gnss} signals may be limited in many indoor environments. They mainly use technologies such as Wi-Fi fingerprinting, Bluetooth beacons, and \ac{uwb} ranging \cite{FurCriBarBar:J21}. In addition, geofencing of \ac{lbs} allows merchants to define virtual boundaries and trigger actions such as location-based messaging, geo-notifications, and different pricing strategies when users enter or leave predefined areas \cite{RodDev:C14}. Additionally, \ac{lbs} plays a vital role in security and emergency response systems, facilitating rapid assistance and enhancing public safety. For example, \cite{HerArtPerOro:J19} proposed a mobile-based intelligent emergency response system to utilize location data to coordinate with emergency services. As for \ac{lbs} recommendation, \cite{SheWanZhuLiu:J22} leverages a consumer's location and behavior through neural networks to provide customized recommendations for various services. Further, \ac{lbs} has revolutionized the food delivery market by allowing taxis to deliver food opportunistically \cite{LiuGuoCheDu:J18}. Optimize the delivery routes and the reward for each participating taxi to improve efficiency and time. On another front, \ac{lbs} social networking platforms permit users to connect with others based on their geographical proximity \cite{WeiQiaSunSun:J22} and \ac{lbs} games integrate gaming elements with real-world positions, offering immersive play experiences for gamers \cite{CheLuLuo:J18,PaaKjeIntTho:C18}. Enabling targeted advertising and promoting, location-based marketing strategies have redefined advertising: \cite{LeWan:J20} examine the effectiveness of location-based mobile advertising in influencing consumer behavior. They highlight the role of contextual offers and social facilitation in enhancing engagement and conversion rates. 

\subsection{GNSS Spoofing Attack}
Spoofing \ac{gnss} involves transmitting fake yet correctly formatted \ac{gnss} signals \cite{HumLedPsiOha:C08,SatStrLenRan:C22,SpaPap:C23}. This deception can cause subtle deviations in timing, signal strength, and arrival angle that are difficult for modern receivers to distinguish. Beyond efforts to authenticate \ac{gnss} signals and messages \cite{FerRijSecSim:J16,Gmv:J21b} to mitigate spoofing attacks, adversaries can still use the recording and replay of authentic \ac{gnss} signals \cite{MaiFraBluEis:C18,LenSpaPap:C22}. Authentication is also not yet supported by the vast majority of \ac{gnss} receivers and requires additional computational resources. A more potent replay/relay attack \cite{ZhaLarPap:J22} similarly employs distance-decreasing attacks to tamper with signals, creating the false impression of earlier arrivals in contrast to their genuine timing. More recent work, such as \cite{SheWonCheChe:C20,GaoLi:J22}, has focused on slowly varying strategies to avoid detection by tightly coupled \ac{gnss}/\ac{imu}. Moreover, \cite{GaoLi:J23} have investigated two time-based gradual spoofing algorithms targeting \ac{gnss} clock information within the navigation message. An overarching subject is the adaptability of attackers who take advantage of versatile \acpl{sdr}. The \acpl{sdr} not only facilitate \ac{gnss} spoofing but additionally pose threats to other alternative navigation infrastructures, including \ac{sbas} and assisted \ac{gps} \cite{ShiDavSag:J20,OhaRusHegAnd:C22}. This complex and diverse threat underscores the urgent need for comprehensive countermeasures to safeguard \ac{gnss} and other \ac{sop} navigation systems.

\subsection{Wireless Network Attack}
Wireless networks provide alternative positioning for \ac{lbs}, yet their broadcast nature makes them susceptible to multiple attacks. Bluetooth replay \cite{VonPeuFraGro:C21} allows an attacker to infiltrate the secured connection from a
victim to another remote one to fake nearby Bluetooth devices. Rogue cellular towers \cite{ShaBorAsoNie:C16} can downgrade \ac{lte} users to less secure and easier to control networks for false positioning services. Among the most potent threats are those targeting Wi-Fi using rogue \acpl{ap} (unauthorized \acpl{ap} installed within a network area) \cite{AloEll:J16}. They can serve as vectors for numerous malicious activities, notably facilitating man-in-the-middle attacks \cite{ThaRifGar:J22}. By intercepting and altering wireless communication from clients, they compromise the integrity and accuracy of channels and localizations. Furthermore, rogue \acpl{ap} can mimic legitimate \acpl{ap} to deceive Wi-Fi clients into automatically connecting to these fraudulent networks, potentially leaking sensitive login credentials or other privacy information \cite{LouPerSta:J23}. Various rogue Wi-Fi attacks (e.g., unfair channel usage, continuous jamming, selective jamming, and channel-based man-in-the-middle) are investigated and implemented using commodity hardware in \cite{VanPie:C14}. Then, the alteration of \acpl{rssi} and related positioning statistics introduces inaccuracies and inconsistencies into \ac{lbs}, undermining their reliability and functionality \cite{FluPotPapHub:C10,YuaHuLiZha:C18}. An open-source implementation for broadcasting Wi-Fi beacons allows arbitrarily manipulation of smartphone Wi-Fi positioning \cite{Skylift2016}. Together, these studies underscore the vital importance of implementing robust countermeasures to safeguard wireless networks and ensure the integrity of \ac{lbs}.

\begin{figure}
\begin{centering}
\includegraphics[trim={0 0 2cm 0},clip,width=\columnwidth]{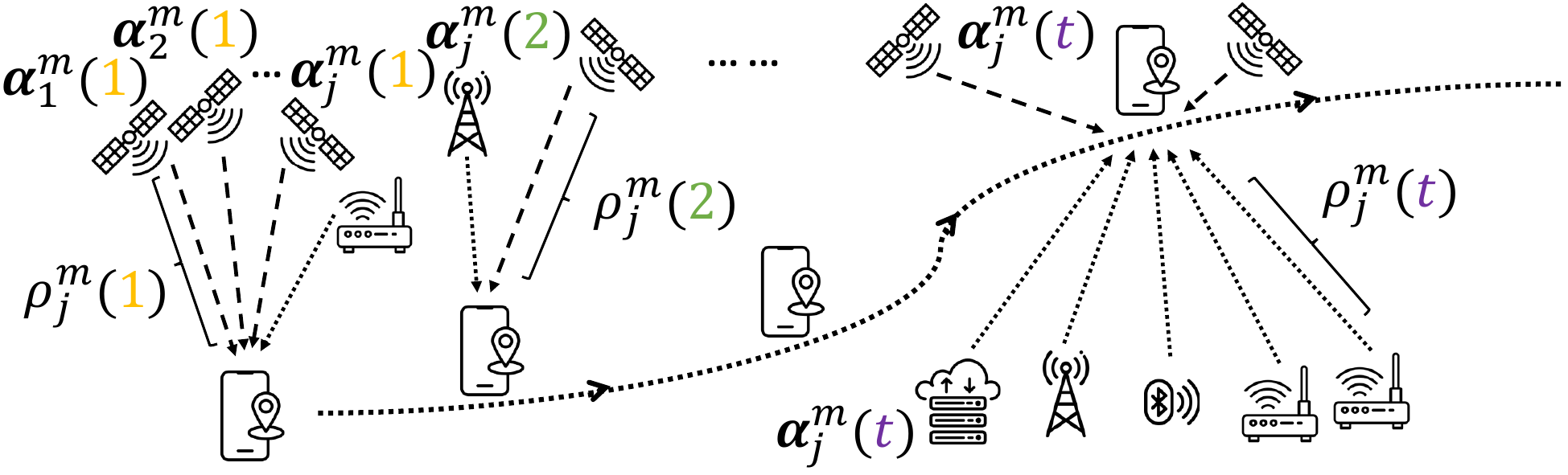}
\par\end{centering}
\caption{\Ac{lbs} applications have ranging information with anchor positions from \ac{gnss} and network infrastructures.}
\label{fig:sysmod}
\end{figure}

\section{System and Adversary Model}
\label{sec:sysmod}
\subsection{System Model}
As shown in Figure~\ref{fig:sysmod}, we consider various \ac{lbs} on a mobile platform that can use both \ac{gnss} and other opportunistic signals for positioning. At time $t$, the platform is located at unknown $\mathbf{p}_{\text{usr}}(t) \in \mathbb{R}^3$. Examples include smartphones, vehicles, and aerial platforms (e.g., drones), which are equipped with modules that leverage opportunistic information to calculate position $\mathbf{p}_{\text{lbs}}(t)$ (an estimate of $\mathbf{p}_{\text{usr}}(t)$) and provide it to the \ac{lbs} application. Opportunistic information includes wireless signals received from network interfaces (e.g., Wi-Fi and cellular networks), Bluetooth, GeoIP, as well as motion data from onboard sensors (e.g., \acpl{imu} and wheel speed sensors).

We define motion measurements as speed $\mathbf{v}(t)$, acceleration $\mathbf{a}(t)$, and orientation $\boldsymbol{\omega}(t)$. Ranging information $\rho_j^m(t)$ is obtained from anchors with known positions $\boldsymbol{\alpha}_j^m(t)$; these anchors include \ac{gnss} satellites, cellular \acpl{bs}, Wi-Fi \acpl{ap}, Bluetooth devices, and GeoIP \ac{rtt} servers. $j \in \mathcal{J}^m(t)$, where $\mathcal{J}^m(t)$ is the set of anchors at time $t$, $m=1,2,...,M$, and $M$ is the number of opportunistic information types. For example, on outdoor \ac{gnss}-enabled phones, ranging information is obtained from \ac{gnss} satellites and terrestrial network infrastructures, including pseudoranges and \acpl{rssi}. Meanwhile, for Wi-Fi devices operating indoors, ranging information is based on \acpl{rssi} collected from Wi-Fi beacons of nearby \acpl{ap}.

However, conditions exist where \ac{gnss} satellites provide incorrect pseudoranges, or other sources (e.g., Wi-Fi) suffer from signal interference. Under benign conditions, position estimations derived from different types of opportunistic information should remain consistent. Hence, if position manipulation occurs, large deviations typically arise between the estimated and actual positions. 

We do not make specific assumptions on the application domain, but it is implied that opportunistic infrastructure is present, possibly in reasonably high numbers and sufficient density, e.g., a deployment of several \acpl{ap} and/or \acpl{bs} within range. This is typical for an urban environment, and our scheme can be used to safeguard commercial \ac{lbs} or civilian protection responders. Deployment in the absence of \acpl{ap} and \acpl{bs} would require other opportunistic positioning sources, e.g., Low Earth Orbit satellites, which are beyond the scope of this work. 

\begin{figure}
\begin{centering}
\includegraphics[trim={0 0 2cm 0},clip,width=\columnwidth]{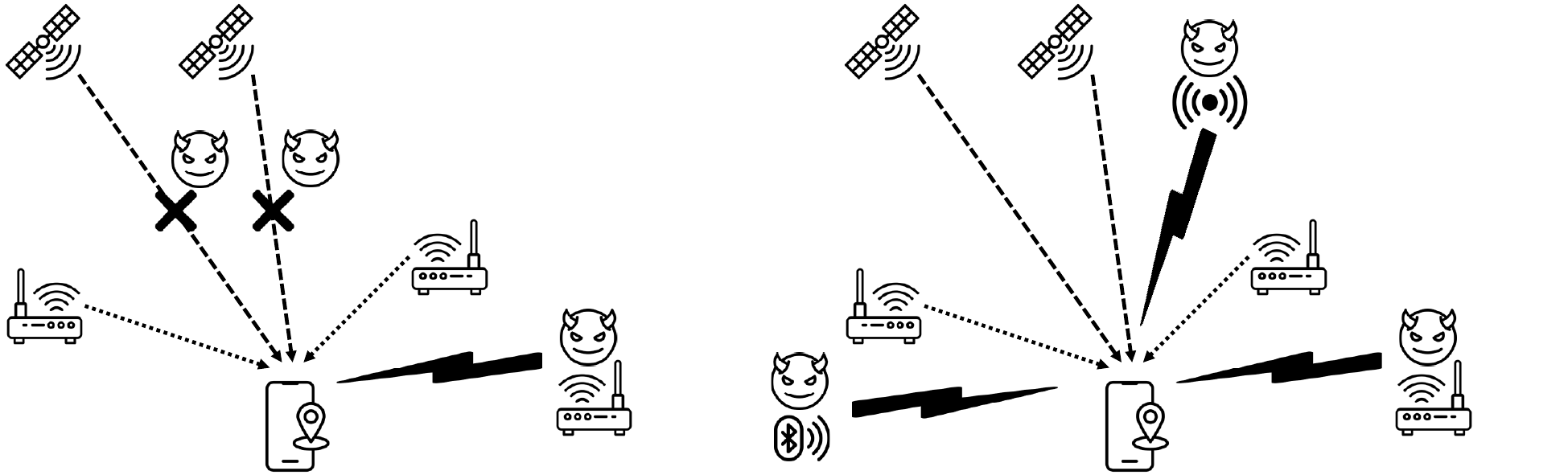}
\par\end{centering}
\caption{Left: Wi-Fi spoofing with \ac{gnss} jamming; right: coordinated location spoofing.}
\label{fig:advers}
\end{figure}

\subsection{Adversary Model}
The adversary can jam, spoof, and replay wireless signals (e.g., \ac{gnss} pseudorandom noise codes, Wi-Fi beacons, cellular signals, Bluetooth beacons), and messages (e.g., \ac{rtt}-related packets) to modify ranging information of the mobile platform. To achieve the modifications, attackers can employ tools such as jammers, spoofers, \ac{gnss} simulators, privacy protection devices \cite{EUSPA:R21}, \acpl{sdr} \cite{LenSpaPap:C22}, or rogue \acpl{ap} \cite{VanPie:C14}. They can take over unsuspecting devices nearby, but cannot alter the algorithm of \ac{lbs} deriving position. The adversary can manipulate the positioning information input and disrupt the integrity and security of \ac{lbs}. Additionally, we assume that the adversary cannot remove all benign ranging information from wireless signals. According to Figure~\ref{fig:advers}, the following attack types are taken into account: Wi-Fi Spoofing with \ac{gnss} Jamming and Coordinated Location Spoofing, detailed in Section~\ref{sec:attlbs}.

\section{Attack Description}
\label{sec:attlbs}
\Ac{lbs} applications play an increasingly crucial role on various mobile platforms. These platforms interact with the external environment via wireless signals. Therefore, this section describes a series of successfully executed \ac{lbs} position manipulation attacks achieved by controlling wireless signals. 

\subsection{Multi-Band GNSS Spoofing}
We begin by presenting a multi-constellation multi-frequency \ac{gnss} spoofing attack, a common form of position manipulation, targeted at Google Maps, which can be generalized to similar map applications or services. 

Method: The attacker first uses a \ac{gnss} jammer to disrupt the victim's reception of legitimate satellite signals, forcing the device to lose its \ac{gnss} lock. Then, the \ac{gnss} simulator generates and broadcasts multi-band \ac{gnss} live-sky signals at a spoofed position chosen by the attacker. The spoofing signals have higher power than the real ones, leading the victim device to compute an attacker-controlled position. 

Impact: Although the platform estimates the position based on a fusion of \ac{gnss}, Wi-Fi, cellular, and Bluetooth, the positioning result will be dominated by \ac{gnss} as long as the \ac{gnss}-provided (possibly attacked) position has a good \ac{dop}. Even if the \ac{gnss} position is not consistent with the network-based positioning result, \ac{lbs} outputs a fused position that is near the spoofed \ac{gnss} position. Moreover, if the victim has enabled the sharing of location data, the spoofed \ac{gnss} position may be incorporated by the location provider to train its network-based positioning database, thus polluting its geolocation system. In addition, the attacker can relay the \ac{gnss} signals from a single location to multiple areas, misleading Google Maps into marking a specific \ac{poi} as congested \cite{EryPap:C22}. 

\begin{figure}
    \centering
    % 0.1cm 5.5cm 0.1cm 4cm
    \includegraphics[trim={0.5cm 5.7cm 0.3cm 6cm},clip,width=.49\columnwidth]{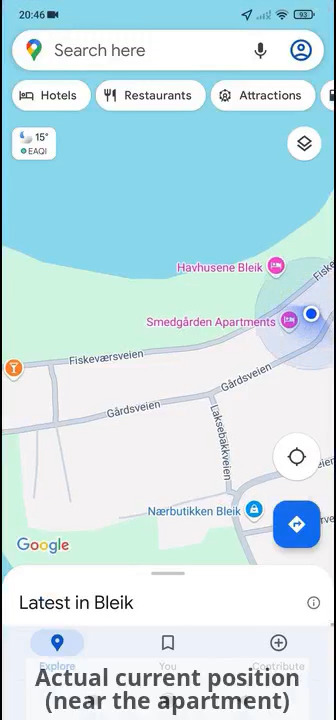}
    \includegraphics[trim={0.5cm 5.7cm 0.3cm 6cm},clip,width=.49\columnwidth]{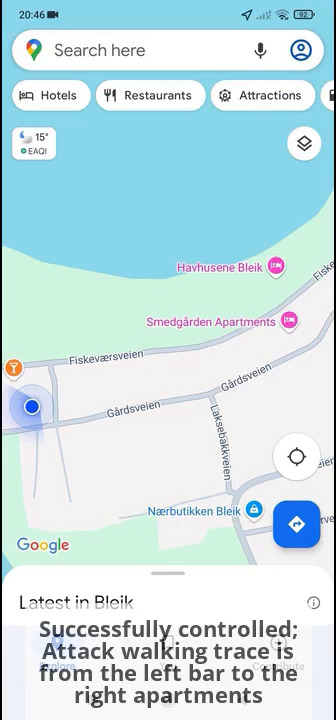}
    \caption{The blue circle in the left screenshot indicates the actual position. The position in the right screenshot is spoofed using a Wi-Fi router with \ac{gnss} jamming. Screenshots are taken from the recorded attack demonstration video.}
    \label{fig:degatt}
\end{figure}

\subsection{Wi-Fi Spoofing with GNSS Jamming}
\label{subsec:attlbsdeg}
Considering high-quality \ac{gnss} simulators cost thousands of dollars, the price of an \ac{sdr} is at least hundreds of dollars, and these instruments are not portable to moving targets (e.g., Uber or DiDi on smartphones), while increasingly practical, the aforementioned \ac{gnss} spoofing attack is not commonly observed in everyday life. However, we then implemented a \emph{Wi-Fi spoofing mixed with \ac{gnss} jamming} at an execution cost under tens of dollars, which simultaneously uses a portable Wi-Fi router and \ac{gnss} jammer to control the position. 

Method: Instead of jamming \ac{gnss} at the beginning, we continuously jam the \ac{gnss} signals throughout the entire duration of the attack. Then, we use a commercial Wi-Fi router (Linksys WRT1200AC) to replay prerecorded Wi-Fi beacons from an intended spoofing trajectory. The Python script of this location spoofing is available below, which should be compatible with Wi-Fi routers installed OpenWrt system and enabled virtual \acpl{ap}. Lastly, the fusion algorithms of \ac{lbs} applications, upon losing \ac{gnss} signals, default to Wi-Fi-based positioning, thereby accepting the spoofed positions. 

Impact: These applications often do not inform the user that \ac{gnss} is not available or when positioning is solely based on network infrastructure. As a result, malicious taxi drivers may manipulate their routes for illegal profits or bypass trip security checks to take the passenger to another destination without triggering any alerts. A demonstration is shown in Figure~\ref{fig:degatt}.
\begin{lstlisting}[language=python, basicstyle=\footnotesize, breaklines=true, caption=Wi-Fi location spoofing script tested on a consumer-grade router installed OpenWrt system and enabled virtual \acpl{ap}. It will broadcast Wi-Fi beacons of \texttt{ap\_set}.]
def uci_set_wifi(ap_set):
    # ap_set is either prerecorded or obtained from WiGLE.net for the targeted spoofing position
    ap_set = ap_set.replace(np.nan, '', regex=True)
    ssh.exec_command("sed -i '37,$d' /etc/config/wireless")
    for i in range(len(ap_set)):
        cmd_to_execute = "uci batch << EOF\n"
        cmd_to_execute += f"""
            set wireless.wifinet{i}=wifi-iface
            set wireless.wifinet{i}.device='radio1'
            set wireless.wifinet{i}.mode='ap'
            set wireless.wifinet{i}.ssid='{ap_set.loc[i, "SSID"]}'
            set wireless.wifinet{i}.encryption='psk2'
            set wireless.wifinet{i}.macaddr='{ap_set.loc[i, "BSSID"]}'
            set wireless.wifinet{i}.key='Password'
            """
        cmd_to_execute += "commit\nEOF"
        ssh.exec_command(cmd_to_execute)
    ssh.exec_command("wifi reload")
\end{lstlisting}

\begin{figure}
\begin{centering}
\includegraphics[trim={0 0.5cm 0 0.45cm},clip,width=\columnwidth]{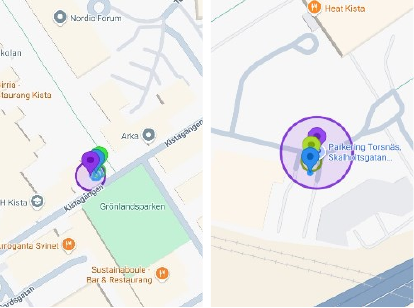}%.7[trim={0.8cm 6.75cm 0.1cm 6.5cm},clip,width=.5\columnwidth]
\par\end{centering}
\caption{Coordinated attack manipulates the network-based (purple circle with pin), \ac{gnss} (green circle with pin), and fused (blue circle with pin) positions of application, deviating from KTH Kista (left) to Skalholtsgatan (right).}
\label{fig:coratt}
\end{figure}

\subsection{Coordinated Location Spoofing}
\label{subsec:attlbscor}
While \ac{gnss} spoofing and \emph{Wi-Fi spoofing with \ac{gnss} jamming} can manipulate positions of most \ac{lbs} applications, we found high-security applications are able to detect these position manipulations. For example, the mobile banking application will locate the device by using multiple position sources (e.g., \ac{gnss}, networks, and IP), then check if they are consistent and within a permitted area. Revolut has location-based security that validates the location of the phone with the location of the offline payment. 

Method: We use a \ac{gnss} simulator alongside a Wi-Fi router (Linksys WRT1200AC) to generate and broadcast \ac{gnss} signals and Wi-Fi beacons (retrieved from public databases \cite{BobArkUht:J23}) at a predetermined spoofed position. In addition, we relay all TCP and UDP packets to a cloud server near the spoofed position using the Wi-Fi router and iptables tool. The iptables configurations are provided below. 

Impact: The attack carefully crafts the wireless signals used for positioning to make them consistent with the spoofed position, demonstrated in Figure~\ref{fig:coratt}. Notably, since Wi-Fi beacons or Bluetooth signals can be crafted, no physical presence near the spoofed location is required to execute this attack.
\begin{lstlisting}[language=bash, basicstyle=\footnotesize, breaklines=true, caption=An example of iptables rules for a proxy server running on the IP \texttt{<TARGET\_IP>} and port \texttt{12345}.]
# Policy Routing Setup
ip rule add fwmark 1 table 100 
ip route add local 0.0.0.0/0 dev lo table 100
# Proxy for LAN Devices
iptables -t mangle -N LAN
iptables -t mangle -A LAN -d 127.0.0.1/32 -j RETURN
iptables -t mangle -A LAN -d 224.0.0.0/4 -j RETURN 
iptables -t mangle -A LAN -d 255.255.255.255/32 -j RETURN 
iptables -t mangle -A LAN -d 192.168.0.0/16 -p tcp -j RETURN
iptables -t mangle -A LAN -d 192.168.0.0/16 -p udp ! --dport 53 -j RETURN
iptables -t mangle -A LAN -j RETURN -m mark --mark 0xff
iptables -t mangle -A LAN -p udp -j TPROXY --on-ip  <TARGET_IP> --on-port 12345 --tproxy-mark 1
iptables -t mangle -A LAN -p tcp -j TPROXY --on-ip  <TARGET_IP> --on-port 12345 --tproxy-mark 1
iptables -t mangle -A PREROUTING -j LAN
\end{lstlisting}

All these demonstrations of the attacks were conducted in an ethical way and did not affect others. The hardware implementation details are provided in Section~\ref{sec:experi}.

\begin{figure}
\begin{centering}
\includegraphics[trim={0 0 0 0},clip,width=\columnwidth]{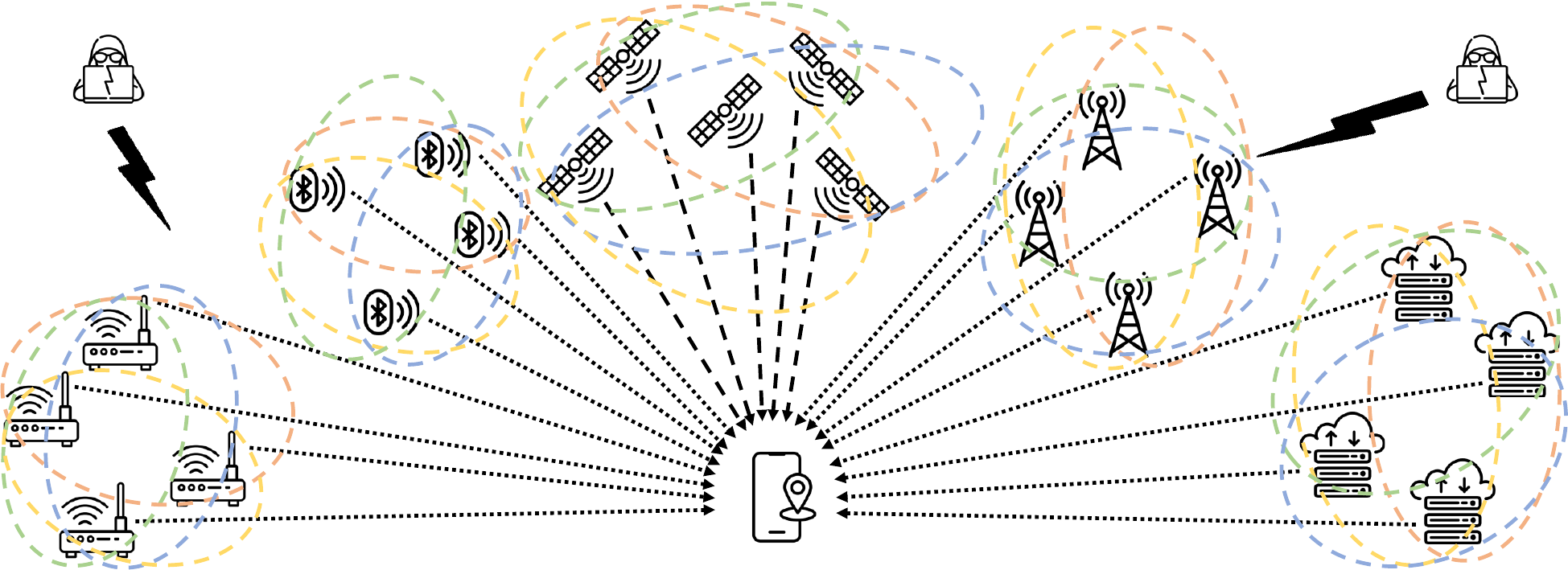}
\par\end{centering}
\caption{Generating subsets of ranging information for cross-validation.}
\label{fig:scheme}
\end{figure}

\section{Defense Scheme}
\label{sec:prosch}
This section describes a defense scheme deployable on mobile platforms like smartphones, utilizing information readily available on such devices. The scheme operates without relying on a trusted location source and considers that any external signal contributing to localization could be manipulated, as demonstrated by the aforementioned attacks.

\subsection{Scheme Outline}
\label{subsec:schout}
We propose an extended \ac{raim} method leveraging opportunistic information to detect \ac{lbs} position manipulation scams. This approach integrates multiple information sources, including \ac{gnss}, \ac{sop} (Wi-Fi, cellular, Bluetooth, GeoIP, etc.), and motion (speed, acceleration, and orientation), to improve detection accuracy. By systematically integrating diverse ranging information into subsets corresponding to heterogeneous infrastructures, the fusion process captures the diverse statistical characteristics (e.g., range, uncertainty, accuracy) of different signal types, as illustrated in Figure~\ref{fig:scheme}. Temporary position estimation for each subset is computed using the corresponding positioning algorithm and then smoothed using movement information from onboard sensors. By comparing the fused position with uncertainty against the original temporary estimations from subsets, the scheme assesses the likelihood that ranging information has been manipulated. Our scheme includes two primary stages: subset generation and position fusion.

In the first stage, real-time ranging information is collected from multiple infrastructures in Section~\ref{subsec:rawdat}. Subsets are created by combining ranging information associated with known anchor positions, in which Wi-Fi \acpl{ssid} can be filtered using an \ac{llm}-based semantics-matching technique in Section~\ref{subsec:llmdat}. We also apply a sampling strategy (Section~\ref{subsec:samstr}) to limit the total number of subsets from Section~\ref{subsec:sizcom}, thereby reducing computational overhead. Finally, temporary position estimates are computed for each subset and infrastructure in Section~\ref{subsec:posmet}. 

In the second stage, onboard sensors collect speed, acceleration, and orientation data to refine temporary position estimates from the first stage. The polynomial regression filtering process in Section~\ref{subsec:motinf} smoothes position estimations with physical constraints derived from the platform's kinematic behavior. The filtered positions are then fused with their positioning uncertainties (Section~\ref{subsec:uncmod}), represented as probability density functions. A composite function (Section~\ref{subsec:likfun}) is derived from these densities to enable probabilistic cross-validation and exclusion of manipulated signals. 

\subsection{Subset Generation}
\label{subsec:subgen}
\subsubsection{Raw Data Preprocessing}
\label{subsec:rawdat}
The input data is similar to other crowd-sourced data used for positioning and \ac{lbs} applications. It opportunistically records information from \ac{gnss} satellites, Wi-Fi \acpl{ap}, cellular \acpl{bs}, Bluetooth devices, GeoIP \ac{rtt} servers, and onboard sensors. \Ac{gnss} signal data includes received times, satellite positions, accuracy, etc. Wi-Fi beacons include \ac{bssid}, \ac{ssid}, and \ac{rssi}. Cellular data includes cell ID, and \ac{rssi}. Bluetooth beacons include \ac{mac} and \ac{rssi}. GeoIP data includes IP address of the platform and \ac{icmp} messages containing \ac{rtt} like ping values to data centers around the world. Onboard sensors provide motion measurements. All of this information has timestamps and aligns with each other type, which is called temporal alignment. 

In \ac{gnss}, pseudorange is ranging information used for positioning. It represents an approximate distance between the satellite and the \ac{gnss} receiver, containing a consistent bias due to receiver clock error. The calculation is the time (the signal takes to reach the receiver) multiplied by the speed of light. Pseudorange error accumulates due to receiver clock drift (caused by inexpensive quartz oscillators), potentially reaching hundreds of meters within seconds without correction. Although the pseudoranges are not the actual distances, the clock is used to measure the pseudoranges at the same time, so all pseudoranges have the same clock error. Then, in the benign environment, this ranging information can still accurately locate the receiver by using measurements from at least four satellites to solve for receiver position (3D) and clock error. Under \ac{gnss} jamming attack, pseudoranges cannot be derived partially or entirely. Under spoofing, some of the pseudoranges are falsified by the attacker, i.e., the derived ranging information deviates from the benign one. 

For Wi-Fi, cellular, and Bluetooth, the received opportunistic ranging information is primarily \ac{rssi}, a negative number in dBm, which follows the log-distance path loss model. The closer the device, the stronger the signal, and vice versa. This ranging information provides a generally less accurate distance approximation than \ac{gnss} pseudoranges. In jamming, the platform cannot receive valid packets, preventing \ac{rssi} derivation. In spoofing, attackers typically use one or more rogue \ac{ap} or \ac{bs} to broadcast signals that alter the set of observed \ac{rssi} values. 

For GeoIP, both the IP address and ranging information from \ac{icmp} utilities are used. \Ac{icmp} utilities (e.g., ping and traceroute) provide time delay (e.g., \ac{rtt}). The delay multiplied by half the speed of light should be an approximation of the distance between the platform and the GeoIP \ac{rtt} server. Delay manipulation attacks, potentially operating at lower network layers, could intercept and forward \ac{icmp} messages to induce artificial delays, leading to inflated range estimates. In the situation of attacks using transparent proxy, this proxy is usually running on the network layer, so only TCP and UDP messages are managed rather than \ac{icmp} messages. Firewalls may be used for dropping or rejecting the \ac{icmp} messages similar to jamming attack, then ranging information cannot be derived from ping values. 

Most importantly, we have a database of anchors, containing positions of Wi-Fi \acpl{ap}, cellular \acpl{bs}, Bluetooth devices, and GeoIP servers, which play a role like \ac{gps} ephemeris. Among all the anchors in the database, we have data cleaning to eliminate incorrect or non-fixed anchors, e.g., personal hotspots, Bluetooth headphones, and public transport Wi-Fi. 

\subsubsection{LLM-based Data Cleaning}
\label{subsec:llmdat}
Our Wi-Fi \ac{ap} positions are from WiGLE.net \cite{BobArkUht:J23}. Since the data is crowd-sourced without a guarantee of quality, we add a process of data cleaning. The process of matching Wi-Fi \ac{ap} \ac{ssid} to their corresponding place names involves leveraging a \ac{llm} to extract relevant text semantic information, looking up \ac{bssid} in the manufacturing list\footnote{\url{https://standards-oui.ieee.org/} or \url{https://macvendors.com/}} to get the company name, and utilizing \ac{poi} APIs for mapping and validating. This process can also be generalized to Bluetooth devices.

The following steps outline the procedure: Step 1 is filtering fixed places. It inputs the \ac{ssid} name of the \ac{ap} using the prompt ``Is this Wi-Fi \ac{ssid} from a static or mobile hotspot: \texttt{SSID\_NAME}? Its manufacturer is \texttt{ORGANIZATION\_NAME}. Please answer static or mobile only.'' This prompt directs the \ac{llm} to distinguish between \acpl{ap} associated with fixed locations and those from mobile devices based on semantics. If \texttt{ORGANIZATION\_NAME} is not available in the manufacturing list, it is replaced by ``not found in IEEE OUI''. By filtering out mobile \acpl{ap}, we focus on identifying fixed places. Step 2 is extracting keywords. It inputs the \ac{ssid} name of the \ac{ap} using the prompt ``Then, can you extract some keywords of the place name from \texttt{SSID\_NAME}? Please answer keywords directly. If not, leave it blank.'' This step prompts the \ac{llm} to extract relevant keywords indicative of the place name associated with the \ac{ap}. The \ac{llm} analyzes the \ac{ssid} to identify keywords that provide insights into the location. Step 3 is searching \ac{poi} API. The extracted keywords from the previous step are used to query a \ac{poi} API. This API searches for relevant places based on the provided keywords. The search result returns a list of potential matches along with their coordinates. The first coordinates in the search result represent the matched location corresponding to the \ac{ap}'s place name.

\subsubsection{Size and Combination}
\label{subsec:sizcom}
Subsets are generated to explore all possible combinations of anchors, from the minimum size required by positioning to the maximum. This process does not assume the number of attacked ranging information (e.g., spoofed pseudoranges, rogue \acpl{ap}, and delayed \acpl{rtt}), guaranteeing to encompass any manipulations. 

For \ac{gnss}, at least four pseudoranges from satellites can determine the receiver position (latitude, longitude, and altitude) and synchronize the receiver clock offset/error, i.e., the size is from four to $J^1(t)$ and then the number of subsets is $\sum_{i=4}^{J^1(t)}C(J^1(t),i)$, where $J^1(t)=|\mathcal{J}^1(t)|$. For \acpl{ap}, \acpl{bs}, or Bluetooth anchors, the receiver position can be determined using at least three \acpl{rssi} in trilateration, and the clock error cannot be estimated. Similarly, for GeoIP, at least three \acpl{rtt} to determine the rough position. Hence, their number of subsets is $\sum_{i=3}^{J^m(t)}C(J^m(t),i)$, where $J^m(t)=|\mathcal{J}^m(t)|,\forall m>1$. These subsets of ranging information (associated with anchor) indexes, $j$ (from $\rho_j^m(t)$), are denoted as $\mathcal{S}_l^m(t)$, where $l=1,2,...,L^m(t)$ and $L^m(t)$ is the total number of subsets for $m$-th infrastructure.

\subsubsection{Sampling Strategy}
\label{subsec:samstr}
Due to the number of generated subsets being huge and the subsets used for localization leading to sizable computational complexity, we employ a subset sampling strategy. 

The simplest one is random sampling, where each subset is selected according to a predefined probability distribution. For example, discrete uniform distribution makes $l$-th subset in $m$-th infrastructure is equally likely to be chosen as each other. It is designed to deal with a variety of subsets without introducing additional bias or skew to $\hat{\mathbf{p}}_{\text{usr}}(t)$. This technique is simple yet robust and shown experimentally to have minimal impact on the cross-validation process. 

Another strategy is greedy \ac{dop} expansion. It searches for a set of the smallest number of ranging information (e.g., satellites) that allows for a valid position solution while ensuring \ac{dop} meets a quality threshold. Then, it iteratively adds a satellite that maintains the \ac{dop} quality until all available satellites are considered. Now we have the first subset for $m$-th infrastructure, $\mathcal{S}_1^m(t)$, and then the process is repeated for the remaining satellites to have $\mathcal{S}_l^m(t),l=2,3,...,L^m(t)$.

Other strategies include mixed-integer (nonlinear) programming that combines subset generation with positioning to separate benign and spoofing signals. Through subset sampling, we ensure that our detection scheme remains powerful and adaptable to heterogeneous opportunistic information sources and attack types. Significantly, it allows for flexible reduction of the computational load. 

\subsubsection{Positioning Methods}
\label{subsec:posmet}
In order to use the heterogeneous ranging information provided by multiple infrastructures, we have various positioning methods for the subsets from them, e.g., trilateration, multilateration, Geolocation APIs, and location fingerprinting. They provide a position estimation together with an uncertainty value. 

Trilateration for \ac{gnss} Single Point Positioning is based on code observations of pseudoranges. The observations are affected by errors such as atmospheric delay, satellite clock errors, and receiver clock errors. GLONASS and \ac{gps} have some differences in the way ionospheric and tropospheric delays are modeled. Additionally, GLONASS uses a different frequency band than \ac{gps}, so the wavelength of the carrier wave is different. The \ac{gnss} observation equations are $\rho_j^m(t)=\lVert\mathbf{p}_{\text{usr}}(t)-\boldsymbol{\alpha}_j^m(t)\rVert+c \delta t_{\text{usr}}+\epsilon_j^m(t)$, where $c \delta t_{\text{usr}}$ is the
receiver clock bias term, and $\epsilon_j^m(t)$ includes atmospheric delays, satellite clock errors, multipath, and noise. Then, the positioning uses the pseudoranges between the receiver and the satellite positions as the main input to calculate the receiver position: $\lVert\hat{\mathbf{p}}_{\text{usr}}(t)-\boldsymbol{\alpha}_j^m(t)\rVert=\rho_j^m(t), \forall j \in \mathcal{S}_l^m(t)$, which can be linearized by Taylor expansion and solved by least squares. The typical accuracy for standalone \ac{gnss} is 5 meters. 

Geolocation APIs are widely used in terrestrial network-based positionings, which propose a weighted least squares problem to minimize the weighted sum of squared distances between anchors and estimated position \cite{Mozilla2023}. These weights are determined based on the inverse square of the ranging information (i.e., \ac{rssi}): 
\begin{equation}
    \underset{\hat{\mathbf{p}}_{\text{usr}}(t)}{\mathop{\min}}\quad \sum_{j \in \mathcal{S}_l^m(t)} \left(\frac{\lVert\hat{\mathbf{p}}_{\text{usr}}(t)-\boldsymbol{\alpha}_j^m(t)\rVert}{\rho_j^m(t)}\right)^2
    \label{eq:posmetgeo}
\end{equation}
which can be solved by SciPy's least squares. The outdoor positioning accuracy ranges from 10 to 100 meters, while for indoor cases, the accuracy is much better, at around 5 meters. 

GeoIP positioning combines tabulation-based and delay-based IP geolocation \cite{AbdMatVan:C17}. Tabulation-based IP geolocation provides a lookup table to map an IP address to an estimated position. Delay-based IP geolocation usually uses 10 to 20 \acpl{rtt} as ranging information with anchors. It first maps \ac{rtt} to distance based on a fitted function from training data. Then, the position is estimated as the centroid of the intersection of circles whose centers are the anchors and radii are the distances. The global positioning accuracy should be tens to hundreds of kilometers. 

Fingerprint-based positioning is quite different from the aforementioned methods. It is popular for indoor Wi-Fi positioning and needs to build a database of fingerprints (mapping \ac{rssi} vectors to known client positions) \cite{AsaGhaSarMul:J22,AsaMagAbd:J25}. Given a subset $\{\rho_j^m(t) \mid j \in \mathcal{S}_l^m(t)\}$ used for positioning, and the \ac{rssi} fingerprint database $\mathcal{T}=\left\{ \{\rho_j^m(t) \mid j \in \mathcal{J}^m(t)\},\mathbf{p}_{\text{usr}}(t) \right\}_{t=T_1}^{T_2}$ pre-surveyed in a benign environment, we find the top $K$ similar fingerprints in $\mathcal{T}$. The similarity scoring function of the $(l,m)$-th subset and \ac{rssi} fingerprint at $t'$ is defined as
\begin{multline}
    f\left(\{\rho_j^m(t) \mid j\in\mathcal{S}_l^m(t)\},\{\rho_j^m(t') \mid j \in \mathcal{J}^m(t')\}\right)\\
    =\sum_{j \in \mathcal{J}^m(t)}\frac{\mathbb{I}\left\{ j \in \mathcal{S}_l^m(t) \right\} }{\max(\left|\rho_j^m(t)-\rho_j^m(t')\right|,d_{\mathrm{min}})}
\end{multline}
where $\mathbb{I}\{\mathrm{A}\}$ takes the value 1 when the condition A is met, and $d_{\mathrm{min}} > 0$ ensures the denominator is greater than 0. Suppose that the positions in $\mathcal{T}$ associated with the most $K$ similar fingerprints are $\mathbf{p}_{\text{usr}}^{(k)},k=1,2,...,K$, with the scores of similarity $f^{(k)},k=1,2,...,K$. Then, the final result of $\hat{\mathbf{p}}_{\text{usr}}(t)$ is the weighted average:
\begin{equation}
    \hat{\mathbf{p}}_{\text{usr}}(t)=\frac{ \sum_{k=1}^K f^{(k)} \mathbf{p}_{\text{usr}}^{(k)} }{ \sum_{k=1}^K f^{(k)} }
    \label{eq:finweiavg}
\end{equation}
with positioning accuracy usually better than 5 meters.

For each $(l,m)$, we then use $\mathbf{p}_l^m(t) \triangleq \hat{\mathbf{p}}_{\text{usr}}(t)$ based on $\mathcal{S}_l^m(t)$ as the subset positioning result.

\subsection{Position Fusion}
\label{subsec:locfus}
\subsubsection{Motion Information}
\label{subsec:motinf}
We use $\mathbf{p}_l^m(t)$ from the previous stage, along with ${\mathbf{v}}(t)$, ${\mathbf{a}}(t)$, $\boldsymbol{\omega}(t)$ from onboard sensors as input data. The motion data allows us to define a kinetic model that describes the platform's state evolution and provides physical constraints for refining positions. We apply local polynomial regression, guided by these constraints, to filter the noise. Positions, $\mathbf{p}_{\text{usr}}(t), \mathbf{p}_{\text{lbs}}(t), \mathbf{p}_l^m(t)$, are represented in the WGS84 format. $\boldsymbol{\omega}(t) \in \mathbb{R}^3$, from onboard sensors, comprises roll ($\phi$), pitch ($\theta$), and yaw ($\psi$) angles defined in the local right-forward-up (RFU) coordinate frame of the mobile platform. The rotation matrix $\mathbf{R}$ converts the local coordinates to WGS84 coordinates \cite{LiuPap:C23}:
\begin{align*}\mathbf{R}(t) & =\mathbf{R}_{\phi}(t)\mathbf{R}_{\theta}(t)\mathbf{R}_{\psi}(t).
\end{align*}
The state of the mobile platform, denoted by $\big(\mathbf{p}_{\text{usr}}(t),\mathbf{v}(t),\mathbf{a}(t) \big)$, evolves from $t-1$ as follows:
\begin{align*}
\mathbf{p}_{\text{usr}}(t)&=\mathbf{p}_{\text{usr}}(t-1)+\mathbf{R}(t-1)\mathbf{v}(t-1)\\
&\qquad+\frac{1}{2}\mathbf{R}(t-1)\mathbf{a}(t-1)+\mathbf {n}\\
\mathbf{v}(t)&=\mathbf{v}(t-1)+\mathbf{a}(t-1)+\mathbf {n}
\end{align*}
where $\mathbf {n}$ models noise. The state transition matrix is
\begin{equation}
    \mathbf {F}(t) ={\begin{bmatrix}\mathbf{1}&\mathbf{R}(t)\\\mathbf{0}&\mathbf{1}\end{bmatrix}}
\end{equation}
and the control-input matrix is $\mathbf {B}(t) =\begin{bmatrix}\mathbf{\frac{1}{2}}\mathbf{R}(t) & \mathbf{1}\end{bmatrix} ^\text{T}$. We denote the estimated state after movement as 
\begin{equation}
    \begin{bmatrix}\overline{\mathbf{p}}_l^m(t)\\\overline{\mathbf{v}}(t)\end{bmatrix}=\mathbf {F}(t-1) \cdot \begin{bmatrix}\mathbf{p}_l^m(t-1)\\\mathbf{v}(t-1)\end{bmatrix}+\mathbf{B}(t-1) \cdot \mathbf{a}(t-1)
    \label{eq:eststamot}
\end{equation}
where $l=1,2,...,L^m(t), m=1,2,...,M$. 

Then, the refined position, denoted $\hat{\mathbf{p}}_l^m(t)$, is determined by solving a local polynomial regression problem for each subset. We use an estimator $\hat{\mathbf{p}}_l^m(t) = \mathbf{W}\mathbf{t}$, where $\mathbf{W} \in \mathbb{R} ^{3 \times (n+1)}$ is a matrix of polynomial coefficients, $n$ is the order of the polynomial regression, $\mathbf{t}$ is a $(n+1)$ dimensional vector and $[\mathbf{t}]_i=t^{i-1}$, and $\mathbf{W}$ at $(l,m,t)$ is calculated from the local polynomial regression problem:
\begin{equation}
    \begin{array}{*{20}{c}}
    {\mathop {\min }\limits_{\mathbf{W}} }&{\sum\limits_{t'=t-w}^{t} K_\text{loc}(t-t') \cdot \lVert\mathbf{W} \mathbf{t'}-\mathbf{p}_l^m(t')\rVert^2} \\ 
    {\text{s.t.}}&{\lVert\mathbf{W} \mathbf{t} - \overline{\mathbf{p}}_l^m(t)\rVert \le \epsilon_\text{t}}
    \end{array}
\label{eq:proall}
\end{equation}
where $w$ is the rolling window size and $K_\text{loc}(t-t')$ is a kernel function assigning a scalar value to ensure the closer data has a higher weight, e.g., $\exp \left(-(t-t')^2\right)$. $\epsilon_\text{t}$ in the constraint represents a small tolerance number for $\overline{\mathbf{p}}_l^m(t)$ to ensure that the estimated position adheres to the movement within a short time duration. Note that $\mathbf{p}_l^m(t')$ may not be available during the entire window $w$; therefore, we also use other verified $\hat{\mathbf{p}}_{\text{usr}}(t')$ to complement missing positions. 

The second derivative of the objective function in \eqref{eq:proall} with respect to $\mathbf{W}$ is $2 \cdot \sum\limits_{t'=t-w}^{t}  K_\text{loc}(t-t')\cdot (\mathbf{t}\cdot \mathbf{t}^\top ) \otimes \mathbb{I}$, where $\otimes$ denotes the Kronecker product. This derivative is always a positive definite matrix. Additionally, since the constraints are affine functions, the problem is convex and solvable in polynomial time using Lagrange multipliers. After solving the problem, we obtain $\hat{\mathbf{p}}_l^m(t)$.

\subsubsection{Uncertainty Modeling}
\label{subsec:uncmod}
$\hat{\boldsymbol{\sigma}}_l^m (t) \in \mathbb{R}^3$ is the uncertainty corresponding to $\hat{\mathbf{p}}_l^m(t)$. 

For trilateration of \ac{gnss}, position \ac{dop} is used as uncertainty. Denote $\mathbf{Q}$ as the covariance matrix of the least squares solution to the navigation equations: 
\begin{equation}
\mathbf{Q} = (\mathbf{A}^\top\mathbf{A})^{-1}
\end{equation}
where $\mathbf{A}$ is the design matrix and is given by:
\begin{equation}
\mathbf{A} = \begin{bmatrix}
\frac{\partial\rho_j^m}{\partial x} & \frac{\partial\rho_j^m}{\partial y} & \frac{\partial\rho_j^m}{\partial z} & 1\\
\vdots & \vdots & \vdots & \vdots
\end{bmatrix},\quad j \in \mathcal{S}_l^m(t)
\end{equation}
and $\frac{\partial\rho_j^m}{\partial x}$, $\frac{\partial\rho_j^m}{\partial y}$, and $\frac{\partial\rho_j^m}{\partial z}$ are the partial derivatives of the pseudorange $\rho_j^m$ with respect to the receiver position coordinates $x$, $y$, and $z$, respectively, in $\hat{\mathbf{p}}_l^m(t)$. Once $\mathbf{Q}$ is obtained, the \ac{dop} can be calculated by taking the square root of the trace of $\mathbf{Q}$:
\begin{equation}
(\sigma_x,\sigma_y,\sigma_z,\sigma_t) \triangleq \sqrt{\text{Tr}(\mathbf{Q})}.
\end{equation}
Then, the uncertainty $\hat{\boldsymbol{\sigma}}_l^m (t)$ is $\sqrt{\sigma_x^2+\sigma_y^2+\sigma_z^2} \cdot \mathbf{1}_3$.

For Geolocation APIs and other least squares algorithms, the uncertainty is represented by the residual of least squares, $\sum_{j \in \mathcal{S}_l^m(t)} \left(\frac{\lVert\hat{\mathbf{p}}_{\text{usr}}(t)-\boldsymbol{\alpha}_j^m(t)\rVert}{\rho_j^m(t)}\right)^2$ in \eqref{eq:posmetgeo}. 

For fingerprint-based positioning, the average of $\left\{f^{(k)}\right\}_{k=1}^K$ in \eqref{eq:finweiavg} serves as the uncertainty of positioning, $\hat{\boldsymbol{\sigma}}_l^m(t)$.

For any other positioning technique, we can use the residual vector of the local polynomial regression to model the uncertainty of the estimated positions.

\subsubsection{Likelihood of Attack}
\label{subsec:likfun}
To detect possible manipulation of positioning data, we denote each temporary position of subsets as a Gaussian distribution with mean $\hat{\mathbf{p}}_l^m(t)$ and standard deviation $\hat{\boldsymbol{\sigma}}_l^m(t)$. This results in the following probability density function:
\begin{equation}
    f_{l,t}^m(\mathbf{p})={\frac {1}{\hat{\boldsymbol{\sigma}}_l^m(t) {\sqrt {2\pi }}}}\exp \left(-{\frac {1}{2}}\left({\frac{\mathbf{p} - \hat{\mathbf{p}}_l^m(t)}{\hat{\boldsymbol{\sigma}}_l^m(t)}}\right)^2\right)
\end{equation}
where the operations are point-wise. 

We then define a likelihood-based consistency score to evaluate how well the \ac{lbs} position, $\mathbf{p}_{\text{lbs}}(t)$, aligns with the fused distribution formed by all temporary positions derived from the subsets: 
\begin{equation}
    f_t(\mathbf{p}_{\text{lbs}}(t))=1-\left(\prod_{m=1}^{M}\left(\prod_{l=1}^{L^m(t)} f_{l,t}^m(\mathbf{p}_{\text{lbs}}(t))\right)^{\frac{1}{L^m(t)}}\right)^{\frac{1}{M}}
\end{equation}
which uses the geometric mean and thus is more sensitive to outliers than the arithmetic mean. A higher value of $f_t(\mathbf{p}_{\text{lbs}}(t))$ implies greater inconsistency between the \ac{lbs} position and the inferred distribution. To decide whether position manipulation is occurring, we predefine a threshold $\Lambda_f$ based on the \ac{roc} curve. An alarm is triggered if the detection score $f_t(\mathbf{p}_{\text{lbs}}(t))$ is greater than $\Lambda_f$. The \ac{roc} curve represents the trade-off between false alarm rates and detection accuracy as a function of $\Lambda_f$. Given a target maximum false alarm rate, the corresponding threshold $\Lambda_f$ is selected accordingly. This process will be shown in our experimental evaluation. 

\subsubsection{Recovering Position}
First, we compute a preliminary fused position $\tilde{\mathbf{p}}(t)$ using a weighted average of all estimated positions:
\begin{equation}
    \tilde{\mathbf{p}}(t)=\frac{\sum_{m=1}^{M}\sum_{l=1}^{L^m(t)}\frac{\hat{\mathbf{p}}_l^m(t)}{\hat{\boldsymbol{\sigma}}_l^m(t)}}{\sum_{m=1}^{M}\sum_{l=1}^{L^m(t)}\frac{1}{\hat{\boldsymbol{\sigma}}_l^m(t)}}
    \label{eq:gauprepos}
\end{equation}
where the operations are also point-wise. Then, we compute the deviation between the $l$-th positioning result of $m$-th infrastructure and the fused position at time $t$ is
\begin{equation}
    d_l^m(t)=\lVert \hat{\mathbf{p}}_l^m(t) - \tilde{\mathbf{p}}(t) \rVert.
\end{equation}
Subsets whose deviation exceeds a statistical threshold $\Lambda_d$ are distinguished as a set that contains spoofed ranging information. The threshold is defined by
\begin{equation}
    \Lambda_d = \mathbb{E}[d_l^m(t)] + n_\Lambda \cdot \sqrt{\mathbb{V}[d_l^m(t)]}
\label{eq:gaumixthr}
\end{equation}
where $n_\Lambda$ is a factor controlling coverage, typically taking the value 3, according to the $3\sigma$ rule of thumb. The subset exclusion process can be performed recursively until the selected subsets no longer change. 

Upon validating all subsets $\mathcal{S}_l^m(t)$ for $m=1,2,...,M$ and $l=1,2,...,L^m(t)$, the validated benign subsets can be used to recover the user's actual position $\mathbf{p}_{\text{usr}}(t)$. 
Denote the set of the index $l$ of benign subsets for $m$-th infrastructure as $\mathcal{L}^m(t)$. Then, the recovered position, $\hat{\mathbf{p}}(t)$, tailored from $\tilde{\mathbf{p}}(t)$ is
\begin{equation}
    \hat{\mathbf{p}}(t)=\frac{\sum_{m=1}^{M}\sum_{l \in \mathcal{L}^m(t)}\frac{\hat{\mathbf{p}}_l^m(t)}{\hat{\boldsymbol{\sigma}}_l^m(t)}}{\sum_{m=1}^{M}\sum_{l \in \mathcal{L}^m(t)}\frac{1}{\hat{\boldsymbol{\sigma}}_l^m(t)}}
\end{equation}
which is together with the attack alarm status. We use $\hat{\mathbf{p}}(t)$ as an alternative to $\mathbf{p}_{\text{lbs}}(t)$ if the likelihood of position manipulation is high. 

\subsection{Theoretical Analysis}
\label{sec:theana}
To theoretically analyze the performance of the proposed \ac{raim}, we set the positioning errors and uncertainties to sufficiently small numbers while preserving attacker-induced deviations and errors \cite{VanBro:C94,LiuPap:C24}. 
Then, the preliminary position in \eqref{eq:gauprepos} becomes
\begin{equation}
    \tilde{\mathbf{p}}(t)=\frac{\sum_{m=1}^{M}\sum_{l=1}^{L^m(t)}\hat{\mathbf{p}}_l^m(t)}{\sum_{m=1}^{M}L^m(t)}
    \label{eq:theprepos}
\end{equation}
and we choose the parameter $n_\Lambda=0$ for the threshold in \eqref{eq:gaumixthr} as $\Lambda_d = \mathbb{E}[d_l^m(t)]$. 
In our context, the geometry of satellite and other infrastructure anchors is assumed to be of good quality for positioning. We consider different numbers of infrastructures in two attack scenarios: (i) coordinated spoofing, where spoofed ranging information is jointly crafted across all infrastructures with a specific spoofing position in mind, and (ii) uncoordinated spoofing otherwise, where spoofed ranging information is independently (and potentially randomly) chosen.
\subsubsection{Single Infrastructure}
To carry out the multilateration algorithm, we denote $N_{\text{min}}$ as the minimum number of anchors required for positioning, e.g., $N_{\text{min}}=4$ for \ac{gnss}. Similarly, $N_{\text{anc}}$ is the total number of anchors (e.g., satellites) providing ranging information, and $N_{\text{adv}}$ is the number of ranging information with deviations caused by the attacker.
\begin{lem}
    Suppose that $N_{\text{anc}}-N_{\text{adv}}>N_{\text{min}}$. Then, $\mathbf{p}_{\text{usr}}(t)$ can be recovered from uncoordinated spoofing.
    \label{lem:uncattrec}
\end{lem}
\begin{proof}
    Under the condition $N_{\text{anc}}-N_{\text{adv}}>N_{\text{min}}$, we obtain the number of benign sets, 
    \begin{equation}
        C(N_{\text{anc}}-N_{\text{adv}}, N_{\text{min}}) > 1
    \end{equation}
    implying that more than one benign position is consistent with each other. Since the spoofed positions from uncoordinated spoofing are random, they can be excluded.
\end{proof}
\begin{prop}
    Suppose that $N_{\text{anc}}-N_{\text{adv}} \ge N_{\text{min}}$. Then, uncoordinated spoofing can be detected. 
\end{prop}
\begin{proof}
    Under the condition $N_{\text{anc}}-N_{\text{adv}} \ge N_{\text{min}}$, we have at least one benign position, which is almost surely inconsistent with a randomly spoofed position. 
\end{proof}
\begin{cor}
    An adversarial subset can contain at most $N_{\text{min}}-1$ benign pseudoranges without being detected.
    \label{cor:advmaxben}
\end{cor}
\begin{proof}
    $N_{\text{min}}$ is determined by unknown variables of the linear system, and it has $N_{\text{min}}$ linearly independent equations. However, the equations for both benign and attack pseudoranges are also linearly independent. $N_{\text{min}}$ benign pseudoranges, when combined with attack pseudoranges, form an overdetermined system.
\end{proof}
\begin{lem}
    Suppose that $\sum_{i=N_{\text{min}}}^{N_{\text{anc}}-N_{\text{adv}}}C(N_{\text{anc}}-N_{\text{adv}},i) > \sum_{i=1}^{N_{\text{adv}}}C(N_{\text{adv}},i)\sum_{j=N_{\text{min}}-i}^{N_{\text{min}}-1}C(N_{\text{anc}}-N_{\text{adv}},j)$.\footnote{In the term $C(N_{\text{anc}}-N_{\text{adv}},j)$, $j>0$ and $j \le N_{\text{anc}}-N_{\text{adv}}$ should hold, or the term takes the value 0.} Then, $\mathbf{p}_{\text{usr}}(t)$ can be recovered from coordinated spoofing.
    \label{lem:corattrec}
\end{lem}
\begin{proof}
    By definition, the number of benign position estimations is $\sum_{i=N_{\text{min}}}^{N_{\text{anc}}-N_{\text{adv}}}C(N_{\text{anc}}-N_{\text{adv}},i)$ and these are consistent with each other. Since an adversarial subset can contain 0 to $N_{\text{min}}-1$ benign pseudoranges or become overdetermined according to Corollary~\ref{cor:advmaxben}, the number of valid spoofed position estimations is $\sum_{i=1}^{N_{\text{adv}}}C(N_{\text{adv}},i)\sum_{j=N_{\text{min}}-i}^{N_{\text{min}}-1}C(N_{\text{anc}}-N_{\text{adv}},j)$. 

    Under the condition $\sum_{i=N_{\text{min}}}^{N_{\text{anc}}-N_{\text{adv}}}C(N_{\text{anc}}-N_{\text{adv}},i) > \sum_{i=1}^{N_{\text{adv}}}C(N_{\text{adv}},i)\sum_{j=N_{\text{min}}-i}^{N_{\text{min}}-1}C(N_{\text{anc}}-N_{\text{adv}},j)$, we have $\tilde{\mathbf{p}}(t)$ in \eqref{eq:theprepos} will be closer to benign estimations than to at least one spoofed estimation. Thus, $\Lambda_d = \mathbb{E}[d_l^m(t)]$ can exclude at least one satellite under attack. Inductively, this leads to fewer adversarial subsets than before, allowing us to continue solving $\tilde{\mathbf{p}}(t)$ to exclude more satellites under attack, and finally to ensure the exclusion of all attacks. 
\end{proof}
\subsubsection{Multiple Infrastructures}
We denote $N_{\text{min}}^m$ as the minimal number of anchors required by the positioning in the $m$-th infrastructure, $N_{\text{anc}}^m$ as the total number of anchors providing ranging information, and $N_{\text{adv}}^m$ as the number of anchors compromised or controlled by the attacker, providing adversarial ranging information.
\begin{thm}
    Suppose that $\left| \left\{ m \mid N_{\text{anc}}^m - N_{\text{adv}}^m - N_{\text{min}}^m = 0 \right\} \right| > 1 \ \lor\ \exists m,\ N_{\text{anc}}^m - N_{\text{adv}}^m - N_{\text{min}}^m > 0$. Then, $\mathbf{p}_{\text{usr}}(t)$ can be recovered from uncoordinated spoofing.
    \label{thm:mulnon}
\end{thm}
\begin{proof}
    $\left| \left\{ m \mid N_{\text{anc}}^m - N_{\text{adv}}^m - N_{\text{min}}^m = 0 \right\} \right| > 1 \ \lor\ \exists m,\ N_{\text{anc}}^m - N_{\text{adv}}^m - N_{\text{min}}^m > 0$ implies either $N_{\text{anc}}^m-N_{\text{adv}}^m=N_{\text{min}}^m$ for at least two different values of $m$, or there exists at least one $m$ for which $N_{\text{anc}}^m - N_{\text{adv}}^m$ is strictly greater than $N_{\text{min}}^m$. Therefore, there exists more than one benign subset. According to Lemma~\ref{lem:uncattrec}, there is more than one benign position among the $M$ infrastructures that can be used for recovering $\mathbf{p}_{\text{usr}}(t)$.
\end{proof}
\begin{thm}
    Suppose that $\sum_{m=1}^M\sum_{i=N_{\text{min}}^m}^{N_{\text{anc}}^m-N_{\text{adv}}^m}C(N_{\text{anc}}^m-N_{\text{adv}}^m,i) > \sum_{m=1}^M\sum_{i=1}^{N_{\text{adv}}^m}C(N_{\text{adv}}^m,i)\sum_{j=N_{\text{min}}^m-i}^{N_{\text{min}}^m-1}C(N_{\text{anc}}^m-N_{\text{adv}}^m,j)$. Then, $\mathbf{p}_{\text{usr}}(t)$ can be recovered from coordinated spoofing.
    \label{thm:mulcol}
\end{thm}
\begin{proof}
    Under the condition $\sum_{m=1}^M\sum_{i=N_{\text{min}}^m}^{N_{\text{anc}}^m-N_{\text{adv}}^m}C(N_{\text{anc}}^m-N_{\text{adv}}^m,i) > \sum_{m=1}^M\sum_{i=1}^{N_{\text{adv}}^m}C(N_{\text{adv}}^m,i)\sum_{j=N_{\text{min}}^m-i}^{N_{\text{min}}^m-1}C(N_{\text{anc}}^m-N_{\text{adv}}^m,j)$, and according to Lemma~\ref{lem:corattrec}, we have $\tilde{\mathbf{p}}(t)$ in \eqref{eq:theprepos} will be closer to the benign estimations. Thus, $\Lambda_d$ can exclude at least one adversary. Inductively, all other adversaries can be excluded. 
\end{proof}
Intuitively, the conditions for multiple infrastructures are more likely to be satisfied compared to a single infrastructure, which makes multi-infrastructure detection and mitigation more robust. As long as at least $N_{\text{min}}^m$ ranging information in one infrastructure is benign, position manipulation scams can be detected. 

\section{Implementation and Evaluation}
\label{sec:experi}
We conduct a comprehensive evaluation of our proposed detection framework against \ac{gnss} spoofing and coordinated attacks in real-world scenarios: measurements collected by the NSS field campaign at Jammertest 2024 \cite{Jam:J24} (we call the part used here the Jammertest 2024 dataset) and measurements collected with an over-the-air testbed for isolated wireless attacks on mobile devices. Given that \ac{gnss} jamming is relatively easy to detect (based on signal power and sky visibility), evaluation of Wi-Fi spoofing mixed with \ac{gnss} jamming is not included. Performance evaluation includes the accuracy of detecting attacks, the accuracy of excluding the ranging information under attack, false alarms, and position recovery. We also show the relations between the performance and the parameters in our scheme so that one can find parameters that meet the requirements. 

\begin{figure}
\centering
\includegraphics[width=\columnwidth]{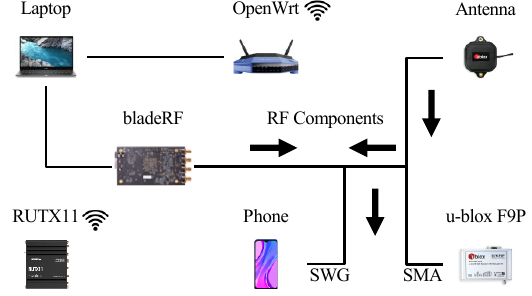}
\caption{High-level architecture of the secure and isolated over-the-air testbed.}
\label{fig:testbed}
\end{figure}

\begin{figure}
\centering
\includegraphics[width=\columnwidth]{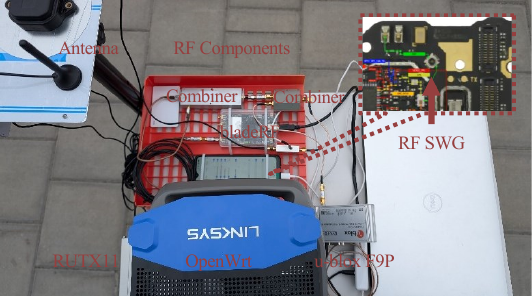}
\caption{Implementation of the secure and isolated over-the-air testbed.}
\label{fig:hardware}
\end{figure}

\subsection{Secure and Isolated Over-the-Air Testbed}
Coordinated attacks on \ac{lbs}, including \ac{gnss} and Wi-Fi spoofing, require a realistic, portable, and over-the-air testbed. Most importantly, it must not interfere with nearby devices or violate spectrum regulations, because conducting such experiments in uncontrolled outdoor settings is not legal. Existing approaches and our previous experiments \cite{LiuPap:C24,LiuPap:C25b} are based on simulation, shielded boxes, or controlled indoor spaces, which limit the realism and test scenarios. Hence, we design and implement a secure and isolated over-the-air testbed with ground truth reference that ensures an ethical, realistic, and outdoor evaluation of \ac{lbs} attacks on real mobile devices, as shown in Figure~\ref{fig:testbed} and Figure~\ref{fig:hardware}. 

The testbed is composed of a testing phone, a bladeRF 2.0 micro xA9, a Linksys WRT1200AC Wi-Fi router installed with OpenWrt, a u-blox EVK-F9P receiver, a u-blox ANN-MB1-00 antenna, a Teltonika RUTX11 4G router, an XPS laptop installed with Windows 11, and RF cables, attenuators, and connectors (LMR195, LMR400, RG316). The testing phone has an SWG-type RF switch connector, and we modified it to connect a cable with an SMA connector. First, the laptop generates \ac{gps} spoofing signals using the bladeGPS tool for real-time signal generation \cite{BladeGPS2022}. Simultaneously, the Wi-Fi router is coordinated to generate beacons to spoof network positioning \cite{LiuPap:C25b}. The ANN-MB1-00 antenna is powered by 5-volt direct current and transmits benign live sky L1+L5 \ac{gnss} signals. Second, benign and attack \ac{gnss} signals are combined with 30 dB attenuation and fed into the testing phone\footnote{Given Wi-Fi operates at unlicensed spectrum, we broadcast them instead of using wired connections.}. Meanwhile, benign \ac{gnss} signals without attenuation are fed into the u-blox receiver as a reference. The receiver is configured to perform precise positioning with the help of our \ac{rtk} station. In addition, both the testing phone and laptop have internet through the 4G router. 

\begin{figure}
    \centering
    \includegraphics[width=\columnwidth]{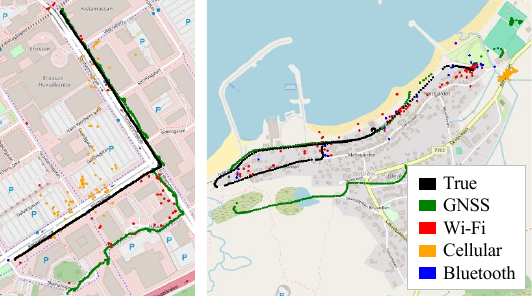}
    \caption{Test sites with two traces: Kista (left) for coordinated spoofing and Bleik (right) for \ac{gnss} spoofing.}
    \label{fig:traces}
\end{figure}

\subsection{Experimental Setup and Data Collection}
We study five walking paths in a 180-meter by 450-meter area in Kista Science City and 68 driving routes in Bleik, a 0.38-square-kilometer village in Norway, as illustrated in Figure~\ref{fig:traces}. Each walking and driving trajectory approximately spans 10--20 minutes. There are many trees and buildings roadside, with heights ranging from 5 to 50 meters, and the buildings in Kista are higher and denser than those in Bleik. For Kista data collection, the raw measurements from \ac{gnss} receiver chip, wireless network signals, and onboard sensors were collected on three Android smartphones: a Redmi 9 testing phone, a Redmi 9 reference phone, and another Google Pixel 8 reference phone. Similarly, in Jammertest 2024, we use six testing phones, i.e., two Pixel 4 XL, two Pixel 8, one Redmi 9, and one Samsung Galaxy S9. 

The u-center software with u-blox ZED-F9P receiver records \ac{gnss} signals for ground truth positions, because it can use benign constellations with a nearby \ac{rtk} reference station. GNSSLogger application records the \ac{gnss} trajectory of Android phone, consisting of \ac{nmea} and \ac{rinex} files, summarized in Table~\ref{tab:rawgnss}. NetworkSurvey application records information on beacons and other messages from network infrastructures near the trajectory, in GeoPackage format, summarized in Table~\ref{tab:rawnet}. In addition, acceleration in meters per second squared, angular velocity in degrees per second, and magnetic field in millitesla are collected from onboard sensors. Note that satellite locations are not derived from \ac{gnss} signals but from precise satellite ephemeris, and the locations of network infrastructures are sourced from WiGLE.net \cite{BobArkUht:J23}. 

\begin{table}
\centering
\caption{Format of the derived major features from \ac{gnss} raw measurements.}
\begin{tabular}{l|l}
\hline
\hline
Time Information & Unix Time [ms] \\
\hline
Distance/Speed & Pseudorange [m] \\
& Pseudorance Uncertainty [m] \\
\hline
Satellite Status & Satellite Identifier \\
& Signal Type \\
& Satellite Location [m,m,m] \\
\hline
\hline
\end{tabular}
\label{tab:rawgnss}
\end{table}

\begin{table}
\centering
\caption{Format of the major features from network raw measurements.}
\begin{tabular}{l|l}
\hline
\hline
Time Information & Unix Time [s] \\
\hline
Ranging Information & Signal Strength [dB] \\
& Signal Frequency [Hz] \\
\hline
Cellular Status & Mobile Country Code \\
& Mobile Network Code \\
& Cell Identifier \\
& Cellular Base Station Location [m,m,m] \\
\hline
Wi-Fi Status & Medium Access Control \\
& Service Set Identifier \\
& Wi-Fi Access Point Location [m,m,m] \\
\hline
Bluetooth Status & Medium Access Control \\
& Service Set Identifier \\
& Bluetooth Device Location [m,m,m] \\
\hline
\hline
\end{tabular}
\label{tab:rawnet}
\end{table}

In Jammertest 2024 \cite{Jam:J24}, the spoofing strategies included stationary spoofing of small/large position jumps, simulated driving, flying spoofing, and more, employing Skydel with two USRP X300 \acpl{sdr} to generate the \ac{gnss} signals following pre-planned routes \cite{Jam:J24}, and \ac{gnss} spoofing is after \ac{gnss} jamming. Although the strategy using a \ac{gnss} jammer and \ac{gnss} spoofer demonstrated a high success rate in the experiments, we found it challenging to seamlessly spoof the smartphone without prior jamming, but rather gradually increasing the power of the spoofer signal and drifting the receiver signal tracking. In the Kista over-the-air testbed experiments, we tested three spoofing strategies: (i) spoofing with enduring \ac{gnss} jamming throughout the entire spoofing period, (ii) short-term temporary jamming until lock on spoofing signals, and (iii) gradual deviating that always controls \ac{gnss} with jamming off. The gradual deviation enduringly spoofs \ac{gnss} position after \ac{gnss} receiver cold-start, but initially makes the spoofing position very close to the actual position, then deviates it from 0 meters to about 150 meters. The whole \ac{gnss} spoofing process takes about 10 minutes. In terms of Wi-Fi spoofing, the rogue \acpl{ap} are placed near the victim to send fake beacons that coordinately pretend to be \acpl{ap} from the spoofing trace, thereby generating 5--15 fake Wi-Fi beacons every second. 

\subsection{Baseline Methods}
\subsubsection{GNSS Spoofing Detection}
We consider the most common detectors and state-of-the-art location providers: Google Play location, network-based detection, Kalman filtering, and location fusion \cite{LiuPap:J25a}. The Google Play location is obtained through the Android API. Network-based positioning performs localization based on network infrastructures and crowd-sourced data, while Kalman filter-based detection estimates a position based on \ac{gnss} position and onboard sensors. The detector then calculates the Euclidean distance between the provided position and the raw \ac{gnss} position. When this distance is larger than a threshold, an attack alarm will be raised. Location fusion \cite{LiuPap:J25a} goes further to securely fuse all opportunistic position information with motion data. It smooths the position estimations from \ac{gnss} and networks with the help of speed, acceleration, etc. Then, a statistical model constructs the uncertainty of these positions. Based on the combined positions with uncertainty ranges, the method decides whether the \ac{gnss} position is within a trusted range. 
\subsubsection{Rogue Wi-Fi AP Detection}
Given that the placement of rogue \acpl{ap}, which broadcasts Wi-Fi signals to attack mobile platforms, inevitably affecting Wi-Fi positioning, we include rogue \acpl{ap} detection methods for comparison. One method is \ac{rssi} clustering implemented and adopted from PRAPD \cite{WuGuDonShi:J18}. It assumes that the Wi-Fi client continuously measures \acpl{rssi} at different times and then uses the $k$-medoid algorithm to cluster these \acpl{rssi} vectors. Ideally, if the environment is stable enough and the \acpl{rssi} vectors are from similar times, the centers of the clusters should be close to each other. Vectors containing \acpl{rssi} from rogue \acpl{ap} are most likely spread out from the cluster center. Another detection method relies on an anomaly detector named ECOD \cite{LiZhaHuBot:J22}. ECOD assumes anomalies are rare and located in the tails of the event's distribution, so anomalies can be identified based on the joint distribution of different dimensions of events. We define the features of the event as the \ac{rssi} differences with respect to \acpl{ap} and times. The detector is hyperparameter-free and calculates the empirical cumulative distribution of these features. When the tail probabilities construct an outlier score much higher than the threshold from the training data, an anomaly is detected. 

\begin{figure}
    \centering
    \includegraphics[width=.86\columnwidth]{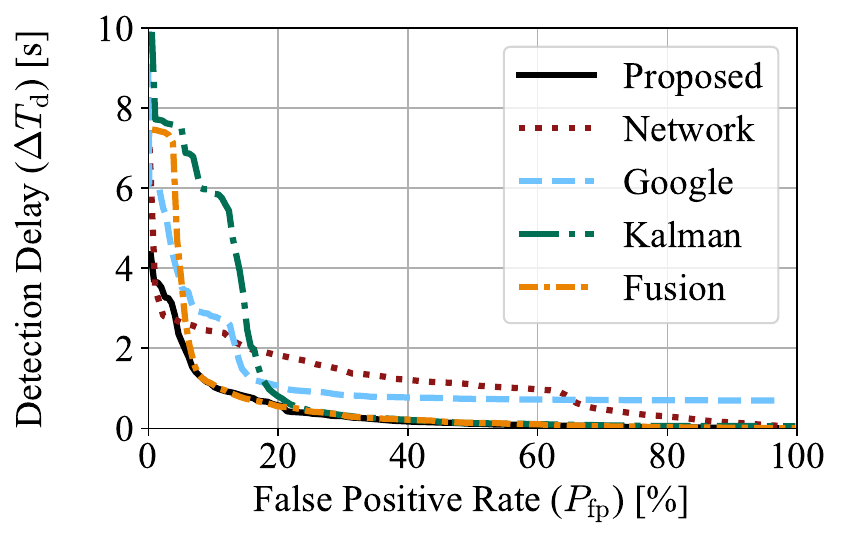}
    \includegraphics[width=.86\columnwidth]{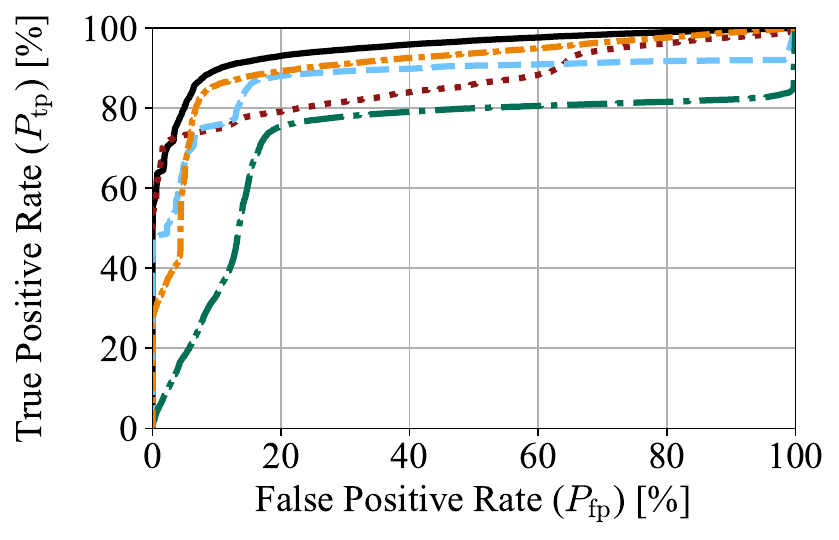}
    \caption{$P_\text{tp}$ and $\Delta T_\text{d}$ of the proposed and baseline methods based on our Jammertest 2024 dataset (for \ac{gnss} spoofing). The same legend applies to both the top and bottom plots.}
    \label{fig:ptptima}
\end{figure}

\subsection{Evaluation Based on Jammertest}
Our evaluation considers the true positive rate ($P_\text{tp}$) versus the false positive rate ($P_\text{fp}$) metric, i.e., \ac{roc} curve, over the baseline methods mentioned above. $P_\text{tp}$ is the percentage of time intervals that attacks are successfully detected over all time intervals under attack, while $P_\text{fp}$ is the percentage of time intervals that are incorrectly classified as under attack over all time intervals not under attack. We choose the experiment parameters $w=15$, $K_\text{loc}(t-t')=\exp \left(-0.3(t-t')^2\right)$, the order of polynomial regression is 2, and a 100\% sampling rate (i.e., all subsets). $\Delta T_\text{d}$ represents the attack detection delay, i.e., the time elapsed between the start of the attack and the moment it is detected.

Figure~\ref{fig:ptptima} presents the results of $P_\text{tp}$ and $\Delta T_\text{d}$ versus $P_\text{fp}$ for Jammertest dataset. Our proposed method improves $P_\text{tp}$ by 9\% to 18\% compared to Google Play location and network-based detection, when $P_\text{fp}$ is between 5\% and 10\%. This confirms that as long as \ac{gnss} \ac{dop} is favorable (\ac{gnss} positions are more precise than those based on network positioning), Google Play location prioritizes \ac{gnss} even when spoofed. Compared with the Kalman filter, the proposed one mostly doubles the performance gain of $P_\text{tp}$. Moreover, our proposed method achieves a $\Delta T_\text{d}$ gain of up to 5 seconds when $P_\text{fp}<10\%$. In particular, $P_\text{tp}$ is evaluated at the level of individual position fixes rather than on entire traces. $P_\text{tp}$ represents the proportion of correctly identified spoofed intervals among all spoofed intervals, while $\Delta T_\text{d}$ measures the detection latency per trace. 

\begin{figure}
    \centering
    \includegraphics[width=.86\columnwidth]{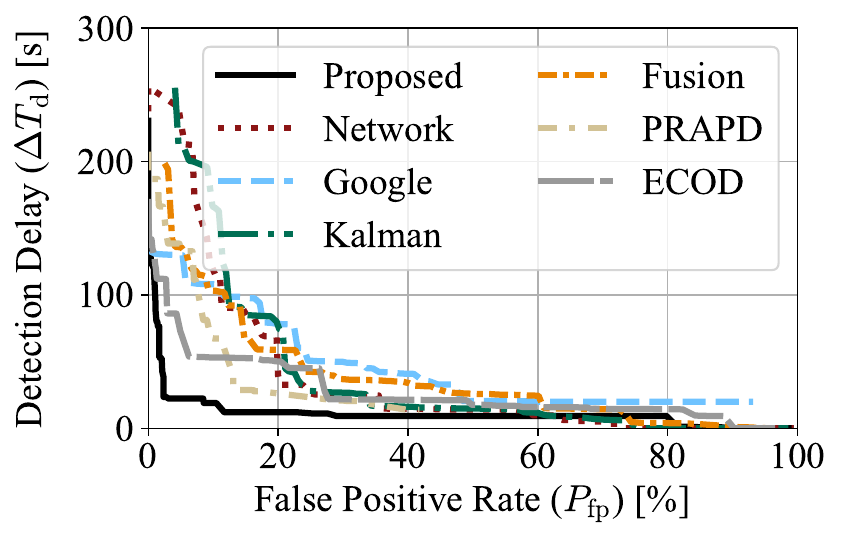}
    \includegraphics[width=.86\columnwidth]{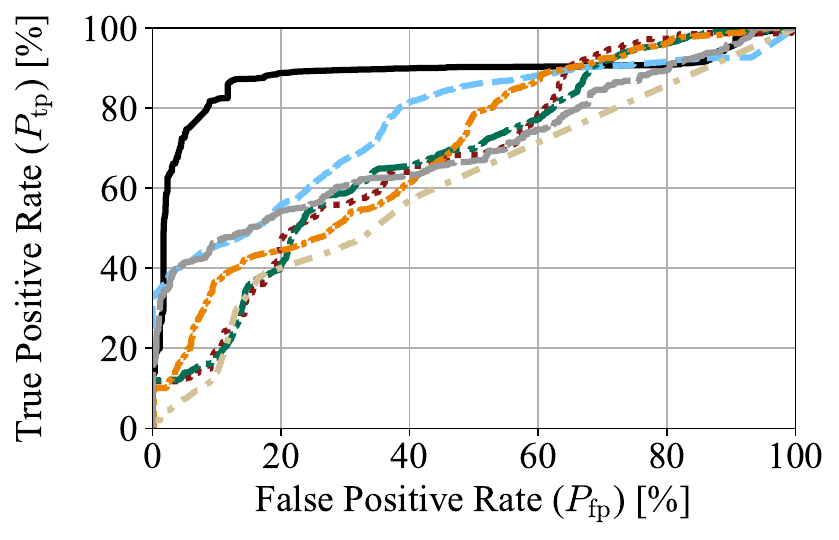}
    \caption{$P_\text{tp}$ and $\Delta T_\text{d}$ of the proposed and baseline methods on the secure and isolated over-the-air testbed (for coordinated attacks). Both plots share the same legend.}
    \label{fig:ptptimb}
\end{figure}

\subsection{Evaluation Based on Testbed}
Our evaluation on the testbed considers the same key metrics as the aforementioned baseline methods: $P_\text{tp}$ and $\Delta T_\text{d}$. 

Recall that we make the spoofed positioning results of \ac{gnss} and Wi-Fi consistent with each other. Then, this coordinated attack is more challenging than \ac{gnss} spoofing, and the results are shown in Figure~\ref{fig:ptptimb}. Our proposed method has a 41\% to 48\% gain of $P_\text{tp}$ compared to the detection based on secure location fusion, when $P_\text{fp}$ is in the range 5\% to 20\%. This is because secure fusion trusts network-based positioning, while Wi-Fi is under spoofing and provides consistent attack results according to \ac{gnss} spoofing. We also observe that Google location provider trusts the network-based positioning less and has better detection in this case, and detection purely based on network-based positioning is not comparable with theirs. As for the Kalman filter, it can hardly detect this kind of spoofing since the deviation is gradually growing. From the perspective of rogue \ac{ap} detection, the proposed method has at least a 50\% gain of $P_\text{tp}$ compared to \ac{rssi} clustering-based detection and more than 32\% improvement over ECOD, with $P_\text{fp}$ in the range 5\% to 20\%. 

The proposed method has a decent performance gain, but the accuracy is slightly lower than in Figure~\ref{fig:ptptima}. This emphasis on coordinated location spoofing is more challenging to detect. 

\begin{table*}
\caption{Absolute error of the positioning when attacked over different methods and datasets.}
\centering
\renewcommand{\arraystretch}{1.3}
\begin{tabular}{l*{8}{c}}
\toprule
\multirow{2}{*}{Methods} & \multicolumn{4}{c}{Jammertest Dataset} & \multicolumn{4}{c}{Testbed Dataset} \\
\cmidrule(r){2-5} \cmidrule(l){6-9}
& Mean & Median & Best 20\% & Worst 20\% & Mean & Median & Best 20\% & Worst 20\% \\
\midrule
Proposed & 253.64 & 75.27 & 10.58 & 358.13 & 108.07 & 33.11 & 7.26 & 201.27 \\
Network & 507.29 & 357.13 & 135.12 & 633.79 & 159.98 & 153.54 & 71.92 & 226.69 \\
Google & 303.38 & 34.16 & 8.97 & 407.55 & 143.34 & 96.37 & 33.13 & 206.58 \\
Kalman & 495.74 & 137.88 & 12.28 & 803.29 & 199.50 & 141.94 & 47.39 & 327.99 \\
Fusion & 379.93 & 146.18 & 16.35 & 543.24 & 124.61 & 37.57 & 16.18 & 200.70 \\
\bottomrule
\end{tabular}
\label{tab:recpos}
\end{table*}

\begin{figure}
    \centering
    \includegraphics[trim={0 0 3cm 0},clip,width=.49\columnwidth]{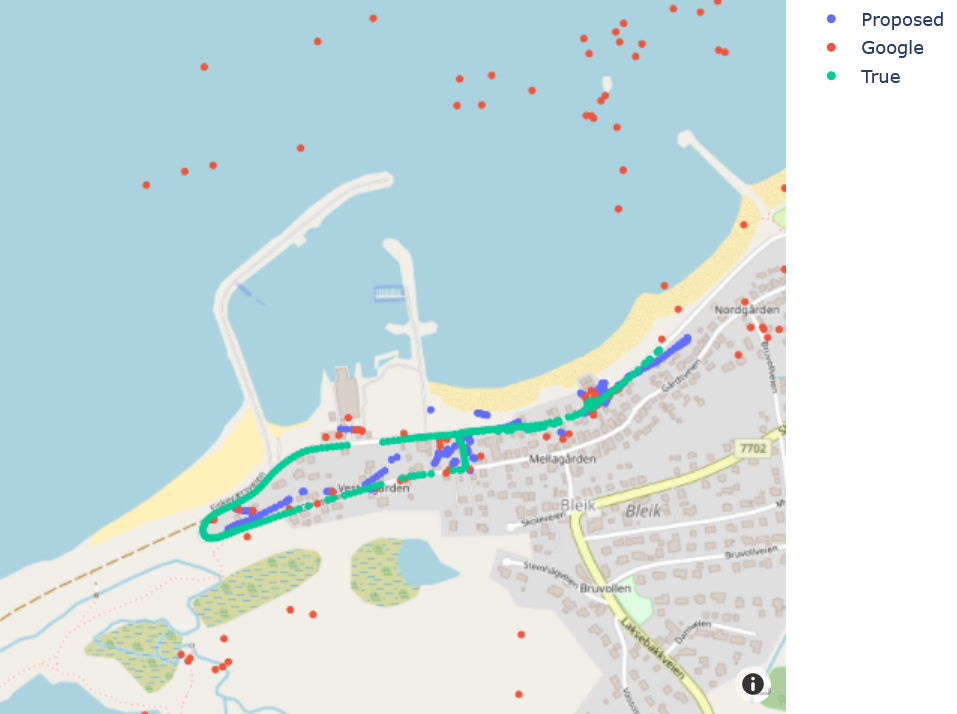}
    \includegraphics[trim={0 0 3cm 0},clip,width=.49\columnwidth]{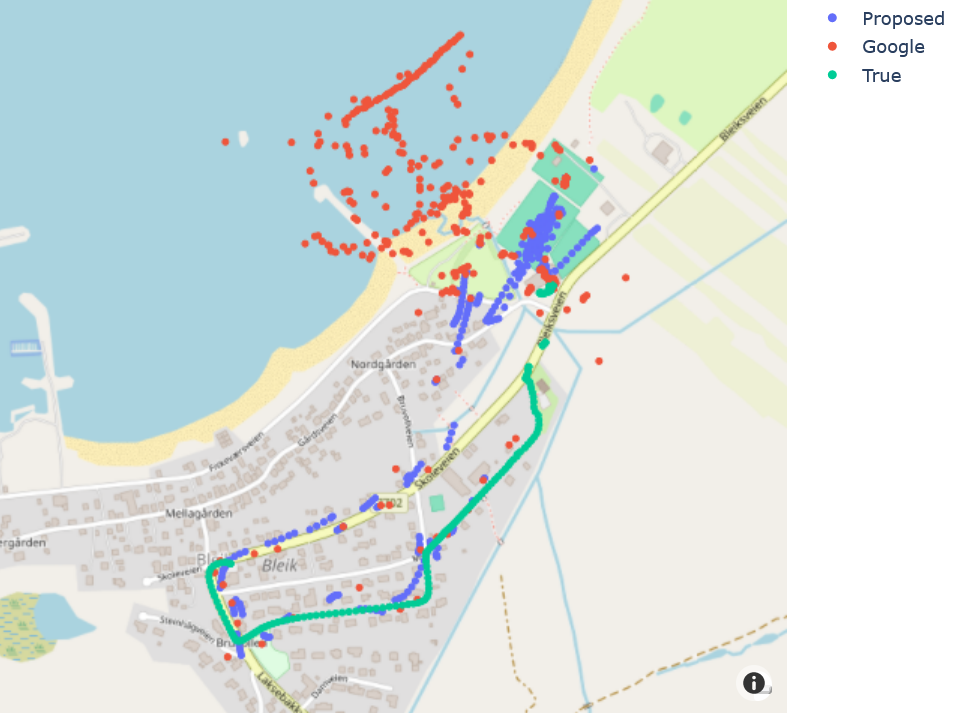}
    \caption{Two trajectories during attack: Cyan dots represent ground truth positions, red dots correspond to Google Play fused positions, and blue dots indicate the recovered positions.}
    \label{fig:apprecpos}
\end{figure}

\subsection{Position Recovery When Attacked}
After the attack is detected and the signals under manipulation are excluded, we can recover a position $\hat{\mathbf{p}}(t)$ that is intended to estimate $\mathbf{p}_{\text{usr}}(t)$ using benign ranging information. We compute the absolute localization error of the recovered position. Table~\ref{tab:recpos} summarizes \ac{mae}, 20th and 80th percentiles of the absolute error distribution: We have 50--254 meters \ac{mae} improvement under \ac{gnss} spoofing of Jammertest dataset, and the proposed position recovery has a median absolute error of 75 meters; when under coordinated location spoofing on the testbed, our method improves \ac{mae} by 16--91 meters, and most absolute error of the proposed position recovery is better than 33 meters. Figure~\ref{fig:apprecpos} presents representative plots illustrating position recovery under location attacks. Considering that our testing datasets are from a consumer-grade \ac{gnss} receiver on a smartphone, which typically has an accuracy of about 10--50 meters in benign environments, the achieved position recovery result is decent. 

\subsection{Ablation Studies}
The proposed detection has different components and parameters; thus, this section analyzes and tunes these parameters based on the Jammertest dataset, given its large scale.
\subsubsection{Effect of Sampling}
Based on the subset sampling strategy proposed in Section~\ref{sec:prosch}, we evaluate $P_\text{tp}$ of detecting \ac{gnss} spoofing and rogue \ac{ap} attacks with different sampling ratios. The comparison of $P_\text{tp}$ with sampling rate is shown in Table~\ref{tab:effsam}. The results indicate that higher sampling rates lead to a monotonic increase in $P_\text{tp}$. When the sampling rate increases from 0.25 to 1.0, $P_\text{tp}$ improves from 77\% to 80\%, at $P_\text{fp}=0.05$. This indicates that higher sampling rates can capture a wider range of subsets, while at lower sampling rates the algorithm may miss important subsets, resulting in reduced accuracy. However, the accuracy is not very sensitive to the sampling rate when $P_\text{fp}\ge0.1$. When the sampling rate decreases from 1 to 0.25 (4 times less computation), $P_\text{tp}$ only decreases by up to 3\%. Although higher sampling rates generally lead to better performance, there may be practical considerations, e.g., computational resources and time constraints.

\begin{table}
% \normalsize
\centering
\caption{$P_\mathrm{tp}$ under different sampling ratios and $P_\mathrm{fp}$.}
\begin{tabular}{c*{5}{c}}\toprule
\diagbox{\textbf{S.}}{$P_\text{fp}$} 
&0.05&0.10&0.15&0.20&0.25\\\midrule
0.25&77.51\% & 87.18\% & 90.87\% & 92.40\% & 93.53\%\\
0.50&78.45\% & 89.15\% & 91.30\% & 92.68\% & 93.59\%\\
0.75&79.51\% & 89.07\% & 91.15\% & 92.55\% & 93.49\%\\
1.00&80.49\% & 89.61\% & 91.66\% & 93.05\% & 93.97\%\\\bottomrule
\end{tabular}
\label{tab:effsam}
\end{table}

\begin{table}
% \normalsize
\centering
\caption{$P_\mathrm{tp}$ under different window sizes and $P_\mathrm{fp}$.}
\begin{tabular}{c*{5}{c}}\toprule
\diagbox{$w$}{$P_\text{fp}$}
&0.05&0.10&0.15&0.20&0.25\\\midrule
5&79.74\% & 87.58\% & 90.10\% & 91.99\% & 93.37\%\\
10&79.73\% & 89.10\% & 91.21\% & 92.70\% & 93.85\%\\
15&80.32\% & 89.49\% & 91.54\% & 92.95\% & 94.01\%\\
20&82.09\% & 90.23\% & 92.39\% & 93.45\% & 94.29\%\\
25&83.04\% & 90.19\% & 92.39\% & 93.57\% & 94.52\%\\
30&82.91\% & 90.29\% & 92.65\% & 93.65\% & 94.46\%\\\bottomrule
\end{tabular}
\label{tab:effwin}
\end{table}

\subsubsection{Effect of Window Size}
Through adjusting $w$ in \eqref{eq:proall}, we provide insights into the detection performance under different window size choices. $w$ ranges from 5 to 30 samples, and other experiment settings are from the coordinated spoofing evaluation. The performance metric $P_\text{tp}$ is shown in Table~\ref{tab:effwin}. Determining the appropriate size requires a trade-off, as a smaller $w$ requires less computing power but may affect the detection accuracy. In contrast, a large $w$ could result in the processing of unnecessary historical data and performance degradation. Our results show an increasing trend of $P_\text{tp}$ for $5<w\leq25$ and become stable when $20<w\leq30$, but larger windows incur significantly higher computational overhead (e.g., processing time for $w=25$ is 1.7 times that of $w=15$). In addition, $P_\text{tp}$ for $w=15$ are slightly better than others and, consequently, it is the most common choice in our detector. We also do a small-scale comparison among different orders of local polynomial regression in $\mathbf{W}$, and $P_\text{tp}$ under the order of polynomial $n=2$ is better than $n=1$. 

\subsection{Computation Overhead}
We evaluated the computational overhead of our proposed scheme on different mobile platforms, e.g., smartphones. These devices typically support logging \ac{gnss} pseudoranges at 1 Hz and network survey data every 3--10 seconds, depending on connectivity. We implemented the algorithm in Python and tested it on Android MTK and Google Tensor devices, observing that each detection cycle can be completed within 1 second, even without specific parallel optimizations, demonstrating feasibility for real-time deployment. Moreover, our design supports the adaptive selection of information sources based on application and computational resources. For security-critical applications, all sensing modules are enabled, including \ac{gnss} and full network data. For less critical applications, only minimal non-sensitive information, such as GeoIP and connected Wi-Fi, is collected. 

The computational complexity of our scheme can be analyzed by breaking down the contributions of its main stages. First, subset generation without sampling for a given infrastructure $m$ is $\mathcal{O}(2^{J^m(t)})$, which grows exponentially with the number of anchors. To mitigate the exponential growth, we employ a sampling strategy that can control the upper-bound number of subsets, $N_\text{sam}$. For each sampled subset, a positioning algorithm is executed, typically with $\mathcal{O}(J^m(t))$ complexity. Second, in position fusion, local polynomial regression for one position is $\mathcal{O}(w)$. Recall that $w$ is the rolling window size. Then, it is $\mathcal{O}(N_\text{sam} \times w)$ for the sampled subsets. As a result, the computational complexity is $\mathcal{O}(N_\text{sam} \sum_{m=1}^M J^m(t)+M \times N_\text{sam} \times w)$ in total, where $M,w$ should be small. Hence, even if the mobile platform computation is not sufficient, the scheme can dynamically adjust the sampling to ensure real-time. 

In our evaluation, the number of \acpl{ap}/\acpl{bs}, $J^m(t)$, varies for the Jammertest 2024 dataset, with few \acpl{ap} and \acpl{bs} available there, and the over-the-air testbed dataset, typically ranging from 10 to 15 for the latter. In such settings, the processing overhead is not high, but if the anchor number were higher, e.g., an order of magnitude higher as it could be in some dense urban settings, processing ranging information from all anchors could be hard, depending on the platform. Nonetheless, the more ranging information from anchors, the more difficult it is for an attacker to spoof the mobile platform position. In addition, one can also preliminarily filter out certain \acpl{ap}/\acpl{bs}, e.g., based on the provider or other trustworthiness, or use sampling strategies, alone or in conjunction with the aforementioned preselection, without validating all the ranging information. 

\section{Related Work}
\label{sec:relwor}
\subsection{RAIM Protecting GNSS}
\Ac{raim} leverages redundant data to cross-validate consistency based on statistics of residuals or subsets of visible satellites \cite{bro:j92}. There are two primary forms of \ac{raim}: residual-based and solution-separation \cite{JoeChaPer:J14}. Residual-based \ac{raim} uses statistical hypothesis checking on residuals to identify potentially inaccurate measurements \cite{KhaRosLanCha:C14,RoyFar:C17}. The residual values can come from the least squares or Kalman filters: \ac{ekf} \ac{raim} makes use of sliding window filters to identify and eliminate outliers using \ac{gps} and inertial sensors \cite{KhaRosLanCha:C14,RoyFar:C17}. Solution-separation \ac{raim} \cite{LiuPap:C24,ZhaPap:J21} recursively assumes faulty satellites, generates subsets of the remaining satellites to derive solutions, and then identifies which subsets contain faults. For example, \cite{ZhaPap:C19} integrates RANSAC clustering to classify position solutions. Further, advanced \ac{raim} \cite{BlaWalEngLee:J15} extends fault exclusion to multiple constellations inclusive of \ac{gps} and Galileo, presenting better integrity compared with single constellation. More recent \ac{sop} techniques \cite{MaaKas:J21,LiuPap:J25a} integrate wireless network infrastructures with \ac{raim} principles. They combine kinematic models, cellular pseudoranges, or Wi-Fi measurements to enhance \ac{raim} overall performance.

\subsection{Rogue Wi-Fi AP Detection}
Detection of malicious Wi-Fi hotspots has received increasing interest. Industry standard solutions \cite{Aru:J24,Ibm:J24,Gan:J24} often employ whitelist-based filtering using known \ac{ap} \ac{mac} addresses and \acpl{ssid} to identify unauthorized \acpl{ap}. However, such records can be easily spoofed using consumer-grade hardware. For example, most of the Wi-Fi routers can set any \ac{ssid}, and some others or open-source routers (e.g., OpenWrt) can even modify their \ac{mac} addresses. Hence, there is a need for detection beyond these Wi-Fi beacons. \cite{AhmAmiKanAbd:C14} uses \ac{rssi} measurements to perceive malicious Wi-Fi signals, but its clustering and two-step algorithm is affected by signal fluctuations as a result of environmental factors (interference and multipath). Another line of defense involves wireless fingerprinting, which operates independently of client devices \cite{LinGaoLiDon:C20}. However, the scalability and robustness of fingerprinting get worse in dynamic network environments (rain, high traffic, etc.). In addition, semantic-based \ac{csi} in \ac{iot} environments offers potential accuracy advantages but requires specialized hardware and large-scale frequency band scanning \cite{BagRoeMarSch:C15}. PRAPD \cite{WuGuDonShi:J18} proposed an \ac{rssi}-based method for moving receivers to get a more general and realistic detection of rogue \ac{ap}. \cite{YanYanYanSon:J22} studied real-time identification of susceptible Wi-Fi connections in operational networks. 

\section{Conclusion}
\label{sec:conclu}
This paper proposed a framework to detect and mitigate coordinated location attacks that manipulate \ac{lbs} positioning data. The detector used opportunistic ranging information and onboard sensors' data. Then, an extended \ac{raim} was proposed to cross-validate the derived position estimates, assess the likelihood of an attack, and recover the actual position when possible. We also demonstrated representative position manipulation attacks in real-world \ac{lbs} applications, designed a testbed to evaluate attacks targeting multiple wireless signals, and showed the feasibility of the proposed scheme with theoretical analysis and experiments in various scenarios. The clearest benefit lies in significantly reducing user and service exposure to scams without requiring additional hardware. Still, this scheme only works for environments with sufficient opportunistic ranging information. %, like in an urban city. For an outdoor open space without networks, one considers integrating low Earth orbit satellites into the framework to detect attacks.

% \section{Acknowledgment}
% This work was supported in part by the KAW Foundation and the China Scholarship Council. We thank the Jammertest 2024 organizers for a live \ac{gnss} jamming/spoofing/meaconing test environment, the former NSS member Dr. Marco Spanghero for support with simulations, smartphone antenna, and circuit components, and the National Academic Infrastructure for Supercomputing in Sweden (NAISS) for computational resources. 

\bibliographystyle{IEEEtran}
\bibliography{reference/references}

\end{document}